\documentclass[times,3p,11pt,sort&compress]{elsarticle}
\pdfoutput=1 
\usepackage[utf8]{inputenc}
\usepackage[T1]{fontenc}
\usepackage{amsmath}
\usepackage{amsthm}
\usepackage{amssymb}
\usepackage{mathtools}
\usepackage{microtype}
\usepackage[textsize=footnotesize,textwidth=1.7cm]{todonotes}
\usepackage{subcaption}
\usepackage{booktabs}
\usepackage{paralist}
\usepackage[ruled,vlined]{algorithm2e}
\usepackage{hyperref}
\usepackage[capitalize,noabbrev]{cleveref}
\newcommand{\fp}{fi\-xed-pa\-ram\-e\-ter}
\usepackage{scrpage2}
\hypersetup{allbordercolors=[HTML]{0404ee},pdfborderstyle={/S/U/W 1}}

\pgfdeclarelayer{background}
\pgfsetlayers{background,main}
\newdefinition{problem}{Problem}
\newtheorem{theorem}{Theorem}
\newtheorem{lemma}{Lemma}
\newtheorem{corollary}{Corollary}

\newdefinition{rrule}{Reduction Rule}
\crefname{rrule}{Reduction Rule}{Reduction Rules}
\crefname{observation}{Observation}{Observations}
\makeatletter
\def\NAT@spacechar{~}%
\makeatother

\newcommand\dego{\text{deg}^{\text{out}}}
\newcommand\degi{\text{deg}^{\text{in}}}

\newcommand\Ni{N^{\text{in}}}
\newcommand\scites{\ensuremath{\mathrm{sum\-Cite}}}
\newcommand\ucites{\ensuremath{\mathrm{un\-ion\-Cite}}}
\newcommand\mcites{\ensuremath{\mathrm{fu\-sion\-Cite}}}
\newcommand{\hind}{H-in\-dex}
\newcommand{\Hind}{H-In\-dex}
\newcommand{\hinds}{H-in\-dices}

\newcommand{\merge}{\ensuremath{\mathcal P}}
\newcommand{\mergelt}{\ensuremath{P}}
\newcommand{\tbpoitsat}{\textsc{3-Bounded Positive 1-in-3-SAT}}
\newcommand{\hindM}{\textsc{\Hind{} Ma\-nip\-u\-la\-tion}}
\newcommand{\hindkM}{\textsc{Cautious \Hind{} Ma\-nip\-u\-la\-tion}}
\newcommand{\hindI}{\textsc{\Hind{} Im\-prove\-ment}}
\newcommand{\pIS}{\textsc{Independent Set}}

\newcommand{\decprob}[3]{%
\noindent
\begin{quotation}
  \noindent
    \begin{minipage}{1.0\linewidth}
      \noindent\textsc{#1}
      \begin{compactdesc}
      \item[Input:] #2
      \item[Question:] #3
      \end{compactdesc}
    \end{minipage}
  \end{quotation}}

\begin{document}
\begin{frontmatter}
  \title{H-Index Manipulation by Merging Articles: Models, Theory, and Experiments\tnoteref{t1}}

  \tnotetext[t1]{An extended abstract of this article appeared at IJCAI~2015 \citep{BKNSW15}.  This version provides full proof details, new kernelization results, as well as additional experiments.}

  \author[nsu,imsoran,tub]{René van Bevern\corref{cor1}}
  \ead{rvb@nsu.ru}
  \cortext[cor1]{Corresponding author}

  \author[jen,tub]{Christian Komusiewicz}
  \ead{christian.komusiewicz@uni-jena.de}

  \author[tub]{Rolf Niedermeier}
  \ead{rolf.niedermeier@tu-berlin.de}

  \author[tub]{Manuel Sorge}
  \ead{manuel.sorge@tu-berlin.de}

  \author[tub,unsw,nicta]{Toby Walsh}
  \ead{toby.walsh@nicta.com.au}

  \address[nsu]{Novosibirsk State University, Novosibirsk, Russian Federation}

  \address[imsoran]{Sobolev Institute of Mathematics, Siberian Branch of the Russian Academy of Sciences, Novosibirsk, Russian Federation}

  \address[jen]{Institut f\"ur Informatik, Friedrich-Schiller-Universität Jena, Germany}

  \address[tub]{Institut f\"ur Softwaretechnik und Theoretische
    Informatik, TU Berlin, Germany}

  \address[unsw]{University of New South Wales, Sydney, Australia}

  \address[nicta]{Data61, Sydney, Australia}

\begin{abstract}
  An author's profile on Google Scholar consists of indexed articles and associated data, such as the number of citations and the \hind{}. The author is allowed to merge articles; this may affect the \hind{}. We~analyze the (parameterized) computational complexity of maximizing the \hind{} using article merges. Herein, to model realistic manipulation scenarios, we define a compatibility graph whose edges correspond to plausible merges. Moreover, we consider several different measures for computing the citation count of a merged article. For the measure used by Google Scholar, we give an algorithm that maximizes the \hind{} in linear time if the compatibility graph has constant-size connected components.  In contrast, if we allow to merge arbitrary articles (that is, for compatibility graphs that are cliques), then already increasing the \hind{} by one is NP-hard. Experiments on Google Scholar profiles of AI researchers show that the \hind{} can be manipulated substantially only if one merges articles with highly dissimilar titles.
\end{abstract}
\end{frontmatter}
\section{Introduction}
\thispagestyle{scrheadings}
\ifoot{\footnotesize\rule{6.5cm}{0.5pt}\\
\parbox{1.0\linewidth}{© 2016.  This manuscript is made available under the CC-BY-NC-ND~4.0 license \url{http://creativecommons.org/licenses/by-nc-nd/4.0/} and was published in \emph{Artificial Intelligence}, doi:10.1016/j.artint.2016.08.001}}
\ihead{}
\ohead{}
\noindent
The \hind{} is a widely used measure for estimating the productivity
and impact of researchers, journals, and institutions. \citet{Hir05}
defined the index as follows: a researcher has \hind{} $h$ if $h$ of
the researcher's articles have at least~$h$ citations and all other
articles have at most~$h$ citations.  Several publicly accessible
databases such as AMiner, Google Scholar, Scopus, and Web of
Science compute the \hind{} of researchers.  Such metrics are
therefore visible to hiring committees and funding agencies when
comparing researchers and proposals.\footnote{Our study on \hind{} manipulation is not meant to endorse or discourage the use of the \hind{} as an evaluation tool. In this regard, we merely aim to raise awareness for the various possibilities for manipulation.}

Although the \hind{} of Google Scholar profiles is computed
automatically, profile owners can still affect their \hind{} by
merging articles in their profile. The intention of providing the
option to merge articles is to enable researchers to identify
different versions of the same article. For example, a researcher may
want to merge a journal version and a version on arXiv.org, which are
found as two different articles by Google's web crawlers. 
This may decrease
a researcher's \hind{} if both articles
counted towards it before merging, or
increase the \hind{} since the merged article may have more
citations than each of the individual articles.  Since the Google Scholar
interface permits to merge arbitrary pairs of articles, this leaves the
\hind{} of Google Scholar profiles vulnerable to manipulation by
insincere authors.

In extreme cases, the merging operation may yield an arbitrarily large \hind{} even if each single article is cited only a few times: If the author has, for example, $h^2$~articles that are cited once, each by a distinct article from another author, then the \hind{} of the profile is~1. Creating $h$~merged articles, each consisting of~$h$ original articles, gives a profile with \hind{}~$h$.  This is the maximum \hind{} achievable with \(h^2\)~citations.

Increasing the \hind{} even by small values could be tempting in particular for young researchers, who are scrutinized more often than established researchers.\footnote{In fact, for senior researchers with many citations, the \hind{} is barely more expressive than the total citation count~\citep{Yon14}.}  \citet{Hir05} estimates that, for the field of physics, the \hind{} of a successful researcher increases by roughly one per year of activity. Hence, an insincere author might try to save years of research work with the push of a few buttons.

\hind{} manipulation by article merging has been studied by \citet{KA13}. In their
model, each article in a profile comes with a number of
citations. Merging two articles, one with~$x$ and one with
$y$~citations, replaces these articles by a new article with~$x+y$
citations. The obtained article may then be merged with further articles to
obtain articles with even higher citation numbers. In this model, one
can determine in polynomial time whether it is possible to improve the
\hind{} by merging, but maximizing the \hind{} by merging is strongly NP-hard~\citep{KA13}.
We extend the results of \citet{KA13} as follows.
\begin{enumerate}
\item We propose two further ways of measuring the number of citations of a merged article. One of them seems to be the measure used by Google Scholar.
\item We propose a model for restricting the set of allowed merge operations. Although Google Scholar allows merges between arbitrary articles, such a restriction is well motivated: An insincere author may try to merge only similar articles in order to conceal the manipulation.
\item We consider the variant of \hind{} manipulation in which only a limited number of merges may be applied in order to achieve a desired \hind{}. This is again motivated by the fact that an insincere author may try to conceal the manipulation by performing only few changes to her or his own profile.
\item We analyze each problem variant presented here within the framework of parameterized computational complexity \citep{CFK+15,DF13,FG06,Nie06}.
That is, we identify parameters~$p$---properties of the input measured in integers---and
aim to design fixed-parameter algorithms, which have running time~$f(p) \cdot n^{O(1)}$ for a computable function~\(f\) independent of the input size~\(n\).  In some cases, this allows us to give efficient algorithms for realistic problem instances despite the NP-hardness of the problems in general.   We also show parameters that presumably cannot lead to fixed-parameter algorithms by showing some problem variants to be \emph{W[1]-hard} for these parameters.
\item We evaluate our theoretical findings by performing experiments with real-world data based on the publication profiles of AI~researchers. In particular, we use profiles of some young and up-and-coming researchers from the 2011 and 2013 editions of the IEEE ``AI's 10 to watch'' list~\cite{ai10tw11,ai10tw13}.
\end{enumerate}

\paragraph{Related work}
Using the models introduced here, Elkind and Pavlou~\cite{EP16}
recently studied manipulation for two alternatives to the \hind{}: the
$i10$-index, the number of articles with at least ten citations, and
the $g$-index~\cite{Egghe2006}, which is the largest number~\(g\)
such that the \(g\)~most-cited
articles are cited at least \(g\)~times
\emph{on average}. They also considered the scenario where merging
articles can influence the profiles of \emph{other} authors.  In a
follow-up work to our findings, we analyzed the complexity of
\emph{unmerging} already merged articles so to manipulate the
\hind{} with respect to the citation measures introduced
here~\cite{BKM+16}. Notably, in the model corresponding to Google
Scholar, the complexity is much lower for unmerging rather than for
merging articles.

A different way of manipulating the \hind{} is by
strategic self-citations \citep{LRT14,Vin13}; \citet{BK11} consider
approaches to detect these.  Strategic self-citations take some effort and are irreversible. Thus, they can permanently damage an author's reputation. In comparison, article merging is easy, reversible and usually justified.

\citet{BK15} showed that, in a previous version of the Google Scholar
interface, which only allowed merges of articles displayed together on one page, it was NP-hard to decide whether a given set of articles can be merged at all. %

The problem of maximizing the \hind{} in the model of \citet{KA13} is
essentially a special case of the scheduling problems~\textsc{Bin
  Covering}~\cite{AJKL84,CCG+13} and \textsc{Machine Covering}~\cite{FD81,WWL16}.

A considerable body of work on manipulation can be found in the
computational social choice literature~\citep{fpaimag10,fhhcacm10}.
If we view citations as articles voting on other articles, then the
problem we consider here is somewhat analogous to strategic
candidacy \citep{djbecon2001}.

\subsection{Our models}\label{sec:mod}
\noindent
We propose two new models for the merging of articles. These models
take into consideration two aspects that are not captured by the model of
\citet{KA13}:
\begin{enumerate}
\item The number of citations of an article resulting from a merge is
  not necessarily the sum of the citations of the merged articles. %
  This is in particular the case for Google
  Scholar.
\item %
  In order to hide manipulation, it would be desirable to
  only merge related articles instead of arbitrary ones. For example, one could only merge articles with similar titles.
\end{enumerate}

\noindent To capture the second aspect, our model allows for
constraints on the compatibility of articles. To capture the first
aspect, we represent citations not by mere citation counts, but using
a directed \emph{citation graph}~$D=(V,A)$. The vertices of~$D$ are the
articles of the author's profile plus the articles that cite them, and
there is an arc~$(u, v)$ in~$D$ if article~$u$ cites article~$v$.

To simplify notation, we assume from now on that we are an author who wants to maximize her or his \hind{} by merging articles. Let~$W\subseteq V$ denote the articles in our profile. In the following, these articles are called \emph{atomic articles} and we aim to maximize our \hind{} by merging some articles in~$W$. The result of a sequence of article merges is a partition~$\merge$ of~$W$. We call each part $P \in \merge$ with $|P| \geq 2$ a \emph{merged article}. Note that having a merged article~$P$ corresponds to performing $|P| - 1$ successive merges on the articles contained in~$P$. It is sometimes convenient to alternate between the partitioning and merging interpretations.

The aim is to find a partition~$\merge$ of $W$ with a
large \hind{}, where the \emph{\hind{} of a partition~$\merge$} is the
largest number~$h$ such that there are at least $h$~parts~$P\in\merge$
whose number~$\mu(P)$ of citations is at least~$h$. Herein, we have
multiple possibilities of defining the measure~$\mu(P)$ of citations of
an article in~$\merge$. Before describing these possibilities, we
introduce some further notation.

\looseness=-1
Let $\degi_D(v)$~denote the indegree of an article~$v$ in the citation graph~$D$, that is, its number of citations.  Analogously, we will use $\dego_D(v)$~to denote the outdegree of an article~\(v\) in the citation graph~\(D\).  Moreover, let~$\Ni_D(v) := \{u \mid (u, v) \in A\}$ denote the set of articles that cite~$v$ and let $\Ni_{D - W}(v) := \{u \mid (u, v) \in A \wedge u \notin W\}$ denote the set of articles that cite~$v$ and are not contained in~$W$ (thus, the articles in $\Ni_{D - W}(v)$ cannot be merged).  For each part~$\mergelt\in\merge$, we consider the following three citation measures for defining the number~$\mu(P)$ of citations of~$P$. They are illustrated in \cref{fig:mergevars}.  The measure
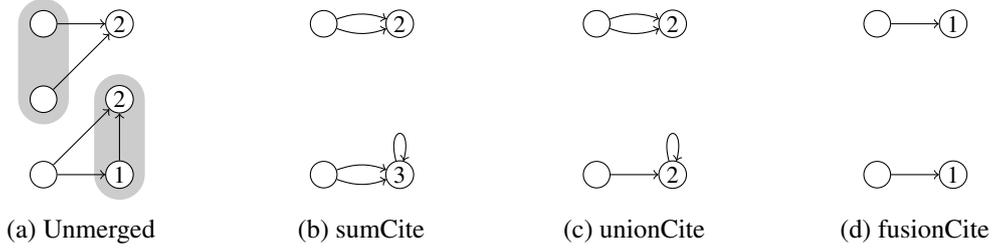
\begin{figure}
  \centering
  \hfill
  \begin{subfigure}{2cm}
    \centering
    \begin{tikzpicture}
      \tikzstyle{edge} = [color=black,opacity=.2,line cap=round, line
      join=round, line width=19pt]

      \tikzstyle{vertex}=[circle,draw,fill=white,minimum
      size=10pt,inner sep=1pt,font=\footnotesize]

      \tikzstyle{citation}=[->] \tikzstyle{mergeable}=[very thick]
      
      \node[vertex] (1) at (0,0) {}; \node[vertex] (2) at (0,1) {};
      \node[vertex] (3) at (0,2) {};

      \node[vertex] (4) at (1,0) {1}; \node[vertex] (5) at (1,1) {2};
      \node[vertex] (6) at (1,2) {2};

      \draw[citation] (1)--(4); \draw[citation] (1)--(5);
      \draw[citation] (2)--(6); \draw[citation] (3)--(6);
      \draw[citation] (4)--(5);
      \begin{pgfonlayer}{background}
        \draw[edge] (2.center)--(3.center); \draw[edge]
        (4.center)--(5.center);
      \end{pgfonlayer}
    \end{tikzpicture}
    \caption{Unmerged}
    \label{fig:unmerged}
  \end{subfigure}
  \hfill
  \begin{subfigure}{1.8cm}
    \centering
    \begin{tikzpicture}
      \tikzstyle{edge} = [color=black,opacity=.2,line cap=round,
      line join=round, line width=19pt]
      
      \tikzstyle{vertex}=[circle,draw,fill=white,minimum
      size=10pt,inner sep=1pt,font=\footnotesize]
      
      \tikzstyle{citation}=[->] \tikzstyle{mergeable}=[very thick]
        
      \node[vertex] (1) at (0,0) {}; \node[vertex] (3) at (0,2) {};
      
      \node[vertex] (4) at (1,0) {3}; \node[vertex] (6) at (1,2)
      {2};
      
      \draw[citation] (1)to[out=20,in=160](4); \draw[citation]
      (1)to[out=-20,in=-160](4); \draw[citation]
      (3)to[out=20,in=160](6); \draw[citation]
      (3)to[out=-20,in=-160](6);

      \draw[citation, loop above] (4)to (4);
      \begin{pgfonlayer}{background}
        \draw[edge,transparent] (2.center)--(3.center); \draw[edge,transparent]
        (4.center)--(5.center);
      \end{pgfonlayer}
    \end{tikzpicture}
    \caption{$\scites$}\label{fig:scites}
  \end{subfigure}\hfill
  \begin{subfigure}{2cm}
    \centering
    \begin{tikzpicture}
      \tikzstyle{edge} = [color=black,opacity=.2,line cap=round, line
      join=round, line width=19pt]
      
      \tikzstyle{vertex}=[circle,draw,fill=white,minimum size=10pt,inner
      sep=1pt,font=\footnotesize]
      
      \tikzstyle{citation}=[->] \tikzstyle{mergeable}=[very thick]
      
      \node[vertex] (1) at (0,0) {}; \node[vertex] (3) at (0,2) {};
      
      \node[vertex] (4) at (1,0) {2}; \node[vertex] (6) at (1,2) {2};
      
      \draw[citation, loop above] (4)to (4);

      \draw[citation] (1)--(4); \draw[citation] (3)to[out=20,in=160](6);
      \draw[citation] (3)to[out=-20,in=-160](6);
      \begin{pgfonlayer}{background}
        \draw[edge,transparent] (2.center)--(3.center); \draw[edge,transparent]
        (4.center)--(5.center);
      \end{pgfonlayer}
    \end{tikzpicture}
    \caption{$\ucites$}
    \label{fig:ucites}
  \end{subfigure}\hfill
  \begin{subfigure}{2cm}
    \centering
    \begin{tikzpicture}
      \tikzstyle{edge} = [color=black,opacity=.2,line cap=round,
      line join=round, line width=19pt]
      
      \tikzstyle{vertex}=[circle,draw,fill=white,minimum
      size=10pt,inner sep=1pt,font=\footnotesize]
      
      \tikzstyle{citation}=[->] \tikzstyle{mergeable}=[very thick]
      
      \node[vertex] (1) at (0,0) {}; \node[vertex] (3) at (0,2) {};
      
      \node[vertex] (4) at (1,0) {1}; \node[vertex] (6) at (1,2)
      {1};
      
      \draw[citation] (1)--(4); \draw[citation] (3)--(6);
      \begin{pgfonlayer}{background}
        \draw[edge,transparent] (2.center)--(3.center); \draw[edge,transparent]
        (4.center)--(5.center);
      \end{pgfonlayer}
    \end{tikzpicture}
    \caption{$\mcites$}
    \label{fig:mcites}
  \end{subfigure}
  \hspace*{\fill}
  \caption{Vertices represent articles in our profile~\(W\), arrows represent citations, numbers are citation counts (note that, in general, there may be articles in~\(V\setminus W\), which are not in our profile and not displayed here).  The articles on a gray background in \subref{fig:unmerged} have been merged in \subref{fig:scites}--\subref{fig:mcites}, and citation counts are given according to the measures $\scites$, $\ucites$, and $\mcites$, respectively.  The arrows represent the citations counted by the corresponding measure.}
  \label{fig:mergevars}
\end{figure}
\begin{align*}
  \scites(\mergelt)&:=\sum_{v\in \mergelt}\degi_D(v)\\
  \intertext{defines the number of citations of a merged article~$P$
    to be the sum of the citations of the atomic articles it contains.
    This is the measure proposed by \citet{KA13}. In contrast, the
    measure}
  \ucites(\mergelt)&:=\Bigl|\bigcup_{v\in \mergelt}\Ni_D(v)\Bigr|\\
  \intertext{defines the number of citations of a merged article~$P$
    as the number of distinct atomic articles citing at least one
    atomic article in~$P$.  We verified empirically that, at the time
    of writing, Google Scholar used the $\ucites$~measure.
    The~measure} \mcites_\merge(\mergelt)&:=\Bigl|\bigcup_{v\in
    \mergelt}\Ni_{D- W}(v)\Bigr| +
  \smashoperator[l]{\sum_{\substack{\mergelt'\in\,\merge\setminus\, \{\mergelt}\}}}
  \begin{cases}
    1&\text{if }\exists v\in\mergelt' \exists w\in\mergelt:(v,w)\in A,\\
    0&\text{otherwise}
  \end{cases}
\end{align*}
is, in our opinion, the most natural one: a set of merged articles is indeed considered to be one article, that is, at most one citation of a
part~$P'\in\merge$ to a part~$P\in\merge$ is counted.
In contrast to
the two other measures, merging two articles under the
$\mcites$~measure may lower the number of citations of the resulting
article and of other articles. Note that, in contrast to \ucites{} and \scites{}, the number of citations of an article according to $\mcites_\merge$ may depend on the partition~\merge. We omit the index $\merge$ where it is clear from the context.  %

To model constraints on permitted article merges, we
consider an undirected \emph{compatibility graph}~$G=(V,E)$. We call
two articles \emph{compatible} if they are adjacent in~$G$. We say
that a partition~$\merge$ of the articles~$W$ \emph{complies} with~$G$
if, for each part~$\mergelt \in \merge$, all articles in~$P$ are
pairwise compatible, that is, if the subgraph~$G[P]$ of~\(G\) induced by~\(P\) is a clique. Thus, if the
compatibility graph~$G$ is a clique, then there are no constraints:
all partitions of~$W$ comply with~$G$ in this case.

Formally, for each measure
$\mu\in\{\scites,\allowbreak\ucites,\allowbreak\mcites\}$, we are
interested in the following problem:

\decprob{\hindM{}($\mu$)}
{A citation graph~$D=(V,A)$, a compatibility graph~$G=(V,E)$, a set $W\subseteq V$ of articles, and a non-negative integer~$h$.}
{Is there a partition of~$W$ that complies with~$G$ and that has \hind{} at least~$h$ with respect to measure~$\mu$?}

\noindent Throughout this work, we use $n:=|V|$ to denote the number of input
articles and $m:=|E|+|A|$ to denote the overall number of edges and
arcs in the two input graphs.  Moreover, we use \(G[S]\)~to denote the subgraph of~\(G\) induced by a subset~\(S\) of its vertices.

\subsection{Our results} 
\noindent
We study the complexity of \hindM{} with
respect to several structural features of the input instances. In particular, we consider the following three parameters:
\begin{itemize}
\item The size~$c$ of the largest connected component in the
  compatibility graph~$G$. We expect this size to be small if only
  reasonable merges are allowed or if all merges have to appear
  reasonable, that is, if compatible articles should superficially look
  similar.
\item The number~$k$ of merges.  %
  An insincere author would hide manipulations using
  a small number of merges.
\item The \hind{} to be achieved.  Although one is interested in
  maximizing the \hind{}, we expect this number also to be relatively
  small, since even experienced researchers seldom have an \hind{} of
  greater than~50.\footnote{More than 99.99 \% of the authors listed at
    \texttt{aminer.org} (accession date 2/27/2016) have an \hind{} of at most 50.}
\end{itemize}
\cref{tab:res} summarizes our theoretical results. For example, we find that,
with respect to the $\ucites{}$ measure used by Google Scholar, it is
easier to manipulate the \hind{} if only a small number of articles 
can be merged into one (small~$c$).
The $\ucites{}$ measure is complex
enough to make increasing the \hind{} by one an NP-hard problem even
if the compatibility graph~$G$ is a clique. In contrast, for the $\scites{}$ measure and the compatibility graph being a clique, it can be decided in
polynomial time whether the \hind{} can be increased by one~\citep{KA13}.

\looseness=-1 We implemented the manipulation algorithms exploiting
small~$k$ and small~$c$.  Experimental results show that all of our
sample AI~authors can increase their \hind{} by only three merges but
that usually merging articles with highly dissimilar titles is
required to obtain a substantial improvement.

\begin{table}\small
  \caption{The complexity of \hindM{} for the citation measures $\scites$, $\ucites$, $\mcites$, and the parameters ``\hind{}~$h$ to achieve'', ``size~$c$ of the largest connected component of the compatibility graph~$G$'', and ``number~$k$ of allowed article merges''. The last row shows the complexity of the variant where we only aim to improve the \hind{} compared to the profile without merges.}
  \centering
  \begin{tabular}{p{0.09cm}p{7.5cm}|p{1.5cm}|p{5.5cm}}
    \toprule
    & $\scites$ & $\ucites$ & $\mcites$\\
    \midrule
    \hspace{-0.1cm}$h$&\multicolumn{3}{c}{%
      W[1]-hard (\cref{thm:hthm}), but linear-time for constant~\(h\) if $G$~is a clique (\cref{lem:kernel-h})}\\
    \midrule
    \hspace{-0.1cm}$c$& \multicolumn{2}{r|}{%
      Solvable in $O(3^c\cdot(n+m))$ time (\cref{lem:small-comp-fpt})}& %
    NP-hard even for~$c=2$ (\cref{qhard})\\
    \midrule
    \hspace{-0.1cm}$k$&%
                        \hfill{}W[1]-hard (\cref{lem:mergesharda}), but solvable in $O(9^kk\cdot (k+n+m))$ time if $G$ is a clique (\cref{lem:mergestrac})&\multicolumn{2}{p{6cm}}{W[1]-hard even if $G$~is a clique (\cref{lem:mergeshardb})}\\
    \midrule[1pt]
    &Improving \hind{} by one is NP-hard (\cref{lem:mergesharda}), but polynomial-time solvable if $G$~is a clique~\citep{KA13}&\multicolumn{2}{p{7.5cm}}{Improving \hind{} by one is NP-hard even if $G$~is a clique (\cref{thm:improvementhard})}\\
    \bottomrule
  \end{tabular}
  \label{tab:res}
\end{table}

\subsection{Preliminaries}
\noindent
We analyze \hindM{} with respect to its
classic and its parameterized complexity. The aim of parameterized complexity
theory is to analyze problem difficulty not only in terms of the
input size, but also with respect to an additional parameter, typically an
integer~$p$ \citep{CFK+15,DF13,FG06,Nie06}. Thus, formally, an instance of a
parameterized problem is a pair~$(I,p)$ consisting of the input~$I$
and the parameter~$p$. A~parameterized problem with parameter~$p$ is
\emph{\fp{} tractable (FPT)} if there is an algorithm that decides
an instance~$(I, p)$ in $f(p) \cdot |I|^{O(1)}$ time, where $f$ is an
arbitrary computable function depending only on~$p$. %
Clearly, if the problem is NP-hard, then we expect $f$~to grow
superpolynomially. %

There are parameterized problems for which there is good evidence that
they are not \fp{} tractable: Analogously to the concept of
NP-hardness, the concept of W[1]-hardness was developed. It is widely
assumed that a W[1]-hard parameterized problem cannot be \fp{} tractable. %
To show that a parameterized problem with parameter~$p'$ is W[1]-hard, a
\emph{parameterized reduction} from a known W[1]-hard parameterized problem with parameter~$p$ can be
used. This is a reduction that runs in~$f(p)\cdot |I|^{O(1)}$ time and produces instances such that the parameter~$p'$ is upper-bounded by
some function~$g(p)$. %
Determining whether an undirected graph~$G$ has a clique of order~$p$ is well known to be W[1]-hard with respect to~$p$.

The notion of a \emph{problem kernel} tries to capture the existence
of efficient and provably effective preprocessing rules \citep{GN07,Kra14}. More
precisely, we say that a parameterized problem has a problem kernel if
every instance can be reduced in polynomial time to an equivalent
instance whose size depends only on the parameter. The algorithm
computing the problem kernel is called \emph{kernelization} and is
often presented as a series of data reduction rules. A data reduction
rule transforms an instance~$(I,p)$ of a parameterized problem into
an~instance~$(I',p')$ of the same problem; a data reduction rule is
\emph{correct} if~$(I,p)$ is a yes-instance if and only if $(I',p')$~is a yes-instance.

\section{Compatibility graphs with small connected components}\label{sec:parc}
\looseness=-1\noindent
In this section, we analyze the parameterized complexity of \hindM{}
parameterized by the size~$c$ of the largest connected component of
the compatibility graph.  This parameterization is motivated by the
fact that one would merge only similar articles and that usually each article is similar to only few other articles.

The following theorem shows that \hindM{} is solvable in linear time
for the citation measures $\scites$ and $\ucites$ if $c$~is
constant. The algorithm exploits that, for these two measures, merging
articles does not affect other articles. Thus, %
we can solve
each connected component independently of the others.

\begin{theorem}{}\label{lem:small-comp-fpt}
  \hindM{}$(\mu)$ is solvable in $O(3^c\cdot(n+m))$~time for
  $\mu\in\{\scites,\allowbreak\ucites\}$ if the connected components
  of the compatibility graph~$G$ have size at most~$c$.
\end{theorem}

\begin{proof}
\looseness=-1  Clearly, articles from different connected components of~$G$
  cannot be together in a part of any partition complying
  with~$G$. Thus, independently for each connected component~$C$ of~$G$,
  we compute a partition of the articles of~$C$ that
  complies with~$G$ and has the maximum number of parts~$P$ with~$\mu(P)\geq h$.

  We first show that this approach is correct and then
  show how to execute it efficiently.  Obviously, if an algorithm
   creates a partition~$\merge$ of the set~$W$ of our own articles
   that complies with~$G$
  and has at least $h$~parts~$P$ with $\mu(P)\geq h$, then we face a
  yes-instance.  Conversely, if the input is a yes-instance, then
  there is a partition~$\merge$ of~$W$ 
  complying with~$G$ and having at least $h$~parts~$P$ with
  $\mu(P)\geq h$.  Consider any connected component $C$ of $G$ and the
  restriction~$\merge_C = \{P \in \merge \mid P \subseteq V(C)\}$
  of~$\merge$ to~$C$, where $V(C)$ is the vertex set of~$C$.  Note that each part in~$\merge$ is either
  contained in $V(C)$ or disjoint from it and, thus, $\merge_C$ is a
  partition of~$V(C)$.  Moreover, merging articles of one connected
  component does not affect the number of citations of articles in
  other connected components with respect to $\scites{}$ or
  $\ucites{}$.  Thus, if we replace the sets of~$\merge_C$ in~$\merge$
  by a partition of~$C$ that has a maximum number of parts~$P$ with
  $\mu(P)\geq h$, then we obtain a partition that still has \hind{} at
  least~$h$.  Thus, our algorithm indeed finds a partition with
  \hind{} at least~$h$.

  We now show how to compute, for each connected component~$C$ of~$G$, a
  partition that maximizes the number of parts with at least
  $h$~citations. In order to achieve a running time of
  $O(3^c\cdot(n+m))$, we employ dynamic programming. First, for
  each connected component~$C$ of~$G$ and every~$V'\subseteq V(C)$, we initialize a table
  \pagebreak[1]
  \[ 
  Q[V']:=
  \begin{cases}
    1&\text{if $G[V']$ is a clique and $\mu(V')\geq h$,}\\
    0&\text{if $G[V']$ is a clique and $\mu(V')<h$, and}\\
    -\infty&\text{otherwise.}
  \end{cases}
  \]
  A table entry~$Q[V']$ thus stores whether merging~$V'$ results in an
  article with at least~$h$ citations.  Obviously, if $G[V']$~is not a
  clique, then $V'$~cannot be a part in any partition complying
  with~$G$.  Therefore, we set~$Q[V']:=-\infty$ in this case.
  All table entries~$Q[V']$ for all vertex
  subsets~$V'$ of all connected components of~$G$ can be computed in $O(2^c \cdot (n+m))$~time.

  Now, for every vertex subset~$V'\subseteq V(C)$ of a connected
  component~$C$, we define~$T[V']$ to be the maximum number of
  parts~$P$ with $\mu(P)\geq h$ in any partition of~$V'$ complying with~\(G\).  Obviously, we have the recurrence relation
  \[
  T[V']=
  \begin{cases}
    0&\text{if $V'=\emptyset$, and}\\
\max_{V''\subsetneq V'}(T[V'\setminus V'']+Q[V''])&\text{otherwise.}
  \end{cases}
  \]
  After computing the table~$Q$, we can compute~$T[V(C)]$ for each connected component~$C$ using dynamic programming: compute \(T[V']\) for each subset~\(V'\subseteq V(C)\) in the order of increasing cardinality.  To this end, for each such subset~\(V'\subseteq V(C)\), we iterate over all subsets~\(V''\subseteq V'\).  Thus, the computation of~\(T\) works in $O(3^cc)$~time since there are at most $3^c$~partitions of~$V(C)$ into $V(C)\setminus(V'\cup V'')$, $V' \setminus V''$, and $V''$.  Herein, the factor~\(c\) accounts for operations with sets of cardinality at most~\(c\).  Thus, the total running time is \(O(2^c\cdot(n+m)+ 3^cc)\subseteq O(3^c\cdot(n+m))\).
\end{proof}

\noindent We have seen that \hindM{} is solvable in
linear time for the citation measures $\scites$ and $\ucites$ if the
compatibility graph has constant-size connected components.  In
contrast, constant-size components of the
compatibility graph do not help %
when the
$\mcites$ measure is used. This we show by a reduction from the NP-hard \tbpoitsat{} problem~\cite{DF09}.

\begin{theorem}\label{qhard}
  \hindM{}$(\mcites)$ is NP-hard even if all of the following conditions hold:
  \begin{enumerate}[i)]
  \item the largest connected component of the compatibility graph has size two,
  \item\label{qhard2} the citation graph is acyclic, and
  \item\label{qhard3} the input instance has \hind{}~\(h-1\).
  \end{enumerate}
\end{theorem}

\noindent Regarding \eqref{qhard2}, note that citation graphs are
often acyclic in practice as papers tend to cite only earlier papers.  Thus,
it is important that \cref{qhard} does not require cycles in the
citation graph.

\begin{proof}We prove Theorem~\ref{qhard} using a polynomial-time many-one reduction from the
  NP-hard \tbpoitsat{} problem \cite{DF09}.

  \decprob{\tbpoitsat}%
  {A formula~$\phi$ in 3-conjunctive normal form containing only
    positive literals and with each literal contained in at most three
    clauses.}%
  {Is there a truth assignment to the variables of~$\phi$ that sets exactly one variable per
    clause to ``true''?}%
  Let
  $n$~be the number of variables of~$\phi$ and let $m$~be the number of
  clauses.  If $m+n$~is odd, then we simply duplicate the instance.  If
  $(m+n)/2<18$, then we solve~$\phi$ using brute force and output a corresponding
  trivial yes- or no-instance of \hindM{}$(\mcites)$. %
  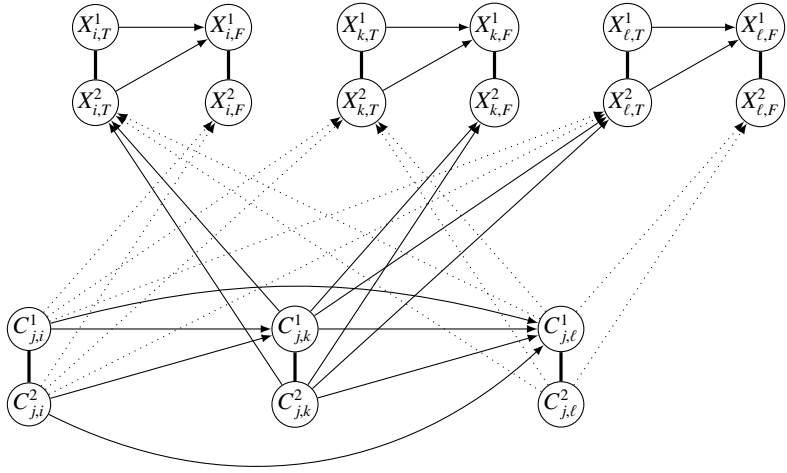
\begin{figure}%
    \centering
    \begin{tikzpicture}[x=1.75cm,y=1cm]
      \tikzstyle{vertex}=[circle,draw,minimum size=4pt,inner sep=0pt,font=\footnotesize]
      \tikzstyle{citation}=[-latex]
      \tikzstyle{mergeable}=[very thick]

      \node[vertex] (x1) at (1,1) {$X_{i,T}^1$};
      \node[vertex] (x2) at (1,0) {$X_{i,T}^2$};
      \node[vertex] (x3) at (2,1) {$X_{i,F}^1$};
      \node[vertex] (x4) at (2,0) {$X_{i,F}^2$};

      \draw[mergeable] (x1)--(x2);
      \draw[mergeable] (x3)--(x4);

      \draw[citation] (x1) --(x3);
      \draw[citation] (x2) --(x3);

      \node[vertex] (y1) at (3,1) {$X_{k,T}^1$};
      \node[vertex] (y2) at (3,0) {$X_{k,T}^2$};
      \node[vertex] (y3) at (4,1) {$X_{k,F}^1$};
      \node[vertex] (y4) at (4,0) {$X_{k,F}^2$};

      \draw[mergeable] (y1)--(y2);
      \draw[mergeable] (y3)--(y4);

      \draw[citation] (y1) --(y3);
      \draw[citation] (y2) --(y3);

      \node[vertex] (z1) at (5,1) {$X_{\ell,T}^1$};
      \node[vertex] (z2) at (5,0) {$X_{\ell,T}^2$};
      \node[vertex] (z3) at (6,1) {$X_{\ell,F}^1$};
      \node[vertex] (z4) at (6,0) {$X_{\ell,F}^2$};

      \draw[mergeable] (z1)--(z2);
      \draw[mergeable] (z3)--(z4);

      \draw[citation] (z1) --(z3);
      \draw[citation] (z2) --(z3);

      \node[vertex] (c1) at (0.5,-3) {$C_{j,i}^1$};
      \node[vertex] (c2) at (0.5,-4) {$C_{j,i}^2$};
      \node[vertex] (c3) at (2.5,-3) {$C_{j,k}^1$};
      \node[vertex] (c4) at (2.5,-4) {$C_{j,k}^2$};
      \node[vertex] (c5) at (4.5,-3) {$C_{j,\ell}^1$};
      \node[vertex] (c6) at (4.5,-4) {$C_{j,\ell}^2$};

      \draw[mergeable] (c1)--(c2);
      \draw[mergeable] (c3)--(c4);
      \draw[mergeable] (c5)--(c6);

      \draw[citation] (c1)--(c3);
      \draw[citation] (c2)--(c3);

      \draw[citation] (c3)--(c5);
      \draw[citation] (c4)--(c5);

      \draw[citation, bend left=15] (c1)to(c5);
      \draw[citation,dotted] (c1)to(y2);
      \draw[citation,dotted] (c2)to(y2);
      \draw[citation,dotted] (c1)to(z2);
      \draw[citation,dotted] (c2)to(z2);

      \draw[citation] (c3)to(x2);
      \draw[citation] (c4)to(x2);
      \draw[citation] (c3)to(z2);
      \draw[citation] (c4)to(z2);

      \draw[citation,dotted] (c5)to(x2);
      \draw[citation,dotted] (c6)to(x2);
      \draw[citation,dotted] (c5)to(y2);
      \draw[citation,dotted] (c6)to(y2);

      \draw[citation, bend right=35] (c2)to(c5);

      \draw[citation,dotted] (c1) to(x4);
      \draw[citation,dotted] (c2) to(x4);

      \draw[citation,dotted] (c5) to(z4);
      \draw[citation,dotted] (c6) to(z4);

      \draw[citation] (c3) to(y4);
      \draw[citation] (c4)to(y4);
    \end{tikzpicture}
    \caption{Construction for a clause~$c_j=(x_i\vee x_k\vee x_\ell)$.
      Undirected bold edges belong to the compatibility graph~$G$,
      directed arcs are citations in the citation graph~$D$.  Only
      vertices in~$W$ and citations between vertices in~$W$ are shown.
      Some arcs are dotted to keep the picture clean.  Observe that
      all citations point from the bottom to the top or from the left
      to the right, and thus form a directed acyclic graph.}
    \label{fig:vargad}
  \end{figure}%
  Otherwise, we now create an instance of \hindM{} with $h:=m+n$.  The
  construction is illustrated in \cref{fig:vargad}.  

  For each
  variable~$x_i$ of~$\phi$, we introduce a variable gadget consisting of
  \begin{itemize}
  \item four articles~$X_{i,T}^1,X_{i,T}^2,X_{i,F}^1$, and
    $X_{i,F}^2$,
  \item two edges~$\{X_{i,T}^1,X_{i,T}^2\}$
    and~$\{X_{i,F}^1,X_{i,F}^2\}$ in the compatibility graph~$G$, and
  \item two arcs~$(X_{i,T}^1,X_{i,F}^1)$
    and~$(X_{i,T}^2,X_{i,F}^1)$ in the citation graph~$D$.
  \end{itemize}
  Merging the pair $\{X_{i,T}^1,X_{i,T}^2\}$ will correspond to
  setting~$x_i$ to true, merging the pair $\{X_{i,F}^1,X_{i,F}^2\}$
  will correspond to setting~$x_i$ to false.  For each
  clause~$c_j=(x_i\vee x_k\vee x_\ell)$, we add a clause gadget
  consisting of
  \begin{itemize}
  \item six articles~$C_{j,z}^1,C_{j,z}^2$ for~$z\in\{i,k,\ell\}$,
  \item three edges~$\{C_{j,z}^1,C_{j,z}^2\}$ for~$z\in\{i,k,\ell\}$ in
    the compatibility graph~$G$, and
  \item six arcs $(C_{j,i}^1,C_{j,k}^1)$, $(C_{j,i}^2,C_{j,k}^1)$,
    $(C_{j,i}^1,C_{j,\ell}^1)$, $(C_{j,i}^2,C_{j,\ell}^1)$,
    $(C_{j,k}^1,C_{j,\ell}^1)$, and $(C_{j,k}^2,C_{j,\ell}^1)$ in the
    citation graph~$D$.
  \end{itemize}
  Merging a pair~$\{C_{j,z}^1,C_{j,z}^2\}$ for $z\in\{i,k,\ell\}$ will
  correspond to setting the literal~$x_z$ of~$c_j$ to true.

  To connect the clause gadget for the clause~$c_j=(x_i\vee x_k\vee
  x_\ell)$ to the corresponding variable gadgets, for all $z\in\{i,k,\ell\}$
  and all~$y\in\{i,k,\ell\}\setminus\{z\}$, we add the arcs
  $(C_{j,z}^1,X_{z,F}^2)$, $(C_{j,z}^2,X_{z,F}^2)$,
  $(C_{j,z}^1,X_{y,T}^2)$, and $(C_{j,z}^2,X_{y,T}^2)$.

  Observe that the constructed citation graph is acyclic since each variable gadget and each clause gadget is acyclic and all other arcs are from clause gadgets to variable gadgets. Moreover, since each variable occurs in at most three clauses of~$\phi$ and each clause has only three variables, every created article has at most $3\cdot 3\cdot 2=18$ incoming citations.  Since $(m+n)/2=h/2\geq 18$, we can finally add, for each of the created articles, a distinct set of articles to~$D$ such that each pair~$\{X_{i,T}^1,X_{i,T}^2\}$, $\{X_{i,F}^1,X_{i,F}^2\}$, or $\{C_{j,z}^1,C_{j,z}^2\}$ is cited exactly $h$~times in total and each single article is cited less than \(h\)~times.  In particular, because articles of the form~\(X_{i,T}^1\) and~\(C_{j,z}^2\) are never cited by any article in variable or clause gadgets, we can add these additional citations so that \(X_{i,T}^1\)~is cited once and \(X_{i,T}^2\)~is cited \(h-1\)~times for each variable~\(x_i\) of~\(\phi\) and so that \(C_{j,z}^2\) in some pair~\(\{C_{j,z}^1,C_{j,z}^2\}\) for a clause~\(c_j\) is cited once, whereas \(C_{j,z}^1\) is cited \(h-1\)~times.   This concludes the construction of our \hindM{}$(\mcites)$ instance.  Note that, for each of the \(h=m+n\)~clauses and variables, we created at least one article with \(h-1\)~citations, which gives an instance with \hind{}~\(h-1\).
  We now prove the correctness of the presented reduction.

  \looseness=-1
  First, if we have an assignment for~$\phi$ that sets exactly one
  variable in each clause to true, then we merge the
  pair~$\{X_{i,T}^1,X_{i,T}^2\}$ for all true variables~$x_i$ and
  merge the pairs~$\{C_{j,i}^1,C_{j,i}^2\}$ for all clauses~$c_j$
  containing~$x_i$.  We will thus get $h=m+n$ articles with
  $h$~citations.
  To show the converse, we first make two important observations:
  \begin{enumerate}
  \item For each variable~$x_i$ of~$\phi$, at most one of the
    pairs~$p_1:=\{X_{i,T}^1,X_{i,T}^2\}$ and
    $p_2:=\{X_{i,F}^1,X_{i,F}^2\}$ can be merged into an
    article~$\mergelt$ with $\mcites(\mergelt)\geq h$: Observe that
    the sum of the citations of each pair is exactly~$h$.  However, if
    both~$p_1$ and~$p_2$ are merged, then the article resulting from
    merging~$p_2$ will get at most~$h-1$ citations: it gets a citation
    from each of the articles of~$p_1$, which will be counted only as
    one citation after merging~$p_1$.
  \item For each clause~$c_j=(x_i\vee x_k\vee x_\ell)$, at most one of
    the pairs~$p_1:=\{C_{j,i}^1,C_{j,i}^2\}$,
    $p_2:=\{C_{j,k}^1,C_{j,k}^2\}$,
    $p_3:=\{C_{j,\ell}^1,C_{j,\ell}^2\}$ can be merged into an
    article~$\mergelt$ with $\mcites(\mergelt)\geq h$: Assume that
    $p_x$ for some~\(x\in\{1,2,3\}\)~is merged into an article~$\mergelt$ with
    $\mcites(\mergelt)\geq h$. Then, no pair $p_y$ with $y > x$ can be
    merged into such an article: $p_y$~can get at most $h-1$~citations
    when merged since $p_y$~gets a citation from each of the two articles
    of~$p_x$.  By the same argument, if any $p_y$ with $y < x$ could be merged into one article~$P$ with $\mcites(\mergelt)\geq h$, then this would contradict the assumption that merging \(p_x\)~yields such an article.
  \end{enumerate}
  Since we ask for merging articles in order to increase the \hind{}
  to $h:=m+n$, which is precisely the number of variables and clauses
  in the input formula, we have to create at least one article with
  $h$~citations for each variable gadget and for each clause gadget.
  That is, if we can achieve \hind{}~$h$, then, for each variable
  gadget and each clause gadget, exactly one pair is merged into an
  article with at least $h$~citations.  Moreover, if, for some
  clause~$c_j=(x_i\vee x_k\vee x_\ell)$ the
  pair~$\{C_{j,z}^1,C_{j,z}^2\}$ is merged for
  some~$z\in\{i,k,\ell\}$, then the pair~$\{X_{z,F}^1,X_{z,F}^2\}$
  cannot be merged into an article with $h$~citations since it gets
  one citation from each of~$C_{j,z}^1$ and~$C_{j,z}^2$.  It follows
  that~$\{X_{z,T}^1,X_{z,T}^2\}$ has to be merged.  Moreover, for
  $y\in\{i,k,\ell\}\setminus\{z\}$, the pair~$\{X_{z,T}^1,X_{z,T}^2\}$
  cannot be merged into an article with $h$~citations, since it gets
  one citation from each of~$C_{j,z}^1$ and~$C_{j,z}^2$.  It follows
  that~$\{X_{z,F}^1,X_{z,F}^2\}$ has to be merged.

  Thus, we obtain an assignment for~$\phi$ that sets exactly one
  variable of each clause to true by setting those variables~$x_i$ to
  true for which the pair~$\{X_{i,T}^1,X_{i,T}^2\}$ is merged into an
  article with at least $h$~citations.
\end{proof}

\section{Merging few articles or increasing the H-Index by one}
\label{sec:park}
\noindent In this section, we consider two variants of \hindM{}:
\hindkM{}, where we allow to merge at most $k$~articles, and \hindI{},
where we ask whether it is possible to increase the \hind{} at all.

\hindkM{} is motivated by the fact that 
insincere authors could try to conceal their
manipulation by merging only few articles. Formally, the problem is defined as follows, where $\mu
\in\{\scites,\allowbreak\ucites,\allowbreak\mcites\}$ as before.
\decprob{\hindkM{}$(\mu)$}
{A citation graph~$D=(V,A)$, a compatibility graph~$G=(V,E)$, a
  set~$W\subseteq V$ of articles, and non-negative integers~$h$ and~$k$.}
{Is there a partition~$\merge$ of $W$ that
  \begin{enumerate}[i)]
  \item complies with~$G$,
  \item has \hind{} at least~$h$ with respect to~$\mu$, and
  \item is such that the
    number~$\sum_{\mergelt\in\merge}(|\mergelt|-1)$ of merges is at
    most $k$?
  \end{enumerate}}

\noindent We show that \hindkM{} parameterized by~$k$ is \fp{} tractable
 only for the $\scites$ measure and when %
 the compatibility graph is a clique. Allowing arbitrary compatibility
 graphs or using more complex measures leads to W[1]-hardness with
 respect to $k$.

The second problem considered in this section, \hindI{}\((\mu)\), %
was introduced by \citet{KA13}; it is formally defined as follows.  \decprob{\hindI{}$(\mu)$}
{A citation graph~$D=(V,A)$, a compatibility graph~$G=(V,E)$, and a
  set~$W\subseteq V$ of articles.}
{Is there a partition~$\merge$ of $W$ that complies with~$G$ and has a
  larger \hind{} with respect to~$\mu$ than the partition of $W$ into singletons?}  

\noindent \cref{qhard}\eqref{qhard3} shows that \hindI($\mcites$) is NP-hard even when the compatibility graph has connected components of size at most two. However, De Keijzer and Apt~\citep{KA13} gave a polynomial-time algorithm for \hindI($\scites$) if the compatibility graph is a clique. We flesh out the boundary between hardness and tractability by proving that more general compatibility graphs lead to NP-hardness in the case of $\scites$ and that even clique compatibility graphs lead to NP-hardness in the case of $\ucites$ and $\mcites$.

\newcommand{\MCC}{\textsc{Multicolored Clique}}

First, we consider the case of general compatibility graphs. Here, we
obtain hardness results for both problem variants by reductions from \MCC{}:

\decprob{\MCC{}}
{An $\ell$-partite undirected graph $H$ along with the $\ell$ partite sets.}
{Is there a clique with $\ell$~vertices contained in~$H$?}

\noindent \MCC{} is known to be NP-hard and W[1]-hard with
respect to~$\ell$ \cite{FDRV09}.  Both hardness results hold
for~$\scites$ and, thus, also for~$\ucites$ and $\mcites$.

\newcommand{\CL}{\textsc{Clique}}
\begin{theorem}{}\label{lem:mergesharda}
 Parameterized by~$k$,   \hindkM{}$(\scites)$ is W[1]-hard. Moreover, \hindI{}$(\scites)$ is NP-hard.
\end{theorem}
\noindent The theorem directly follows from the following lemma, from which we will derive another hardness result in the next section.

\begin{lemma}\label{corlem}
  There is a polynomial-time many-one reduction from \MCC{} to
  \hindkM{}$(\scites)$ with \(k=\ell-1\) and \(h=\ell\) and to
\hindI{}.
\end{lemma}

  \begin{proof}
    The two reductions from the \MCC{} problem differ only in specifying the \hind{} that we want to achieve and the upper bound on the number of merges for \hindkM{}.  The reductions work as follows. %
    We create a citation graph~$D$, a compatibility graph~$G$, and a set of articles~$W$, such that the instance~$(D, G, W, \allowbreak h := \ell, k := \ell-1)$ of \hindkM{} and the instance~$(D, G, W)$ of \hindI{} are yes-instances if and only if $(H,\ell)$~is a yes-instance for \MCC{}.

    Our \hindkM{} and \hindI{} instances have an article
    set~$W=W_\geq\uplus W_<$, where $W_<:=V(H)$ and $W_\geq$~consists
    of $\ell - 1$~new articles. For each article $w \in W_\geq$ we
    introduce a set of $\ell$~articles that are not contained in $W$
    and that cite $w$ and no other article. Similarly, for each
    article~$w \in W_<$ we introduce one article not in~$W$ that
    cites~$w$ and no other article. In this way, we have implicitly
    defined the citation graph~$D$. Next, we construct the
    compatibility graph $G$ from~$H$ by adding each article in~$W_\geq$ as an independent vertex. This concludes the
    construction. Clearly, we can carry it out in polynomial
    time. Note that the reduction is a parameterized
    reduction from \MCC{} parameterized by $\ell$ to \hindkM{} parameterized by~$k$ since~$k=\ell-1$.

    \looseness=-1
    Now we prove the equivalence of the three instances. If the \MCC{}
    instance $(H, \ell)$ is a yes-instance, then there is a clique $S$
    of size $\ell$ in $H$. Merging the corresponding articles $S
    \subseteq W_<$ complies with the compatibility graph and, hence,
    yields a merged article with $\ell$~citations. Together with the
    $\ell - 1$~articles in~$W_\geq$, this results in $\ell$~articles
    with $\ell$~citations and, hence, \hind{} at least $\ell =
    h$. Furthermore, exactly $\ell - 1$~merges are performed in this
    way, implying that the \hindkM{} instance is a yes-instance.

    Note that the \hind{} of the singleton partition $\mathcal{W}$ of
    $W$ is $\ell - 1$. That is, the \hindkM{} instance asks to
    increase the \hind{} of~$\mathcal{W}$ by one. Thus, clearly, if
    the \hindkM{} instance is a yes-instance, then also the \hindI{}
    instance is.

    Finally, assume that the \hindI{} instance is yes. Then there is a
    merged article~$S$ with $\ell$~citations. Since only articles
    in~$W_<$ can be merged, $S$~consists of at least
    $\ell$~articles. Furthermore, $G[S] = H[S]$ is a clique since the
    merging has to comply with~$G$. Hence, the \MCC{} instance is yes,
    concluding the proof that all three instances are equivalent.
  \end{proof}

  \noindent Now we restrict the compatibility graph to be a
  clique, meaning that arbitrary pairs of articles can be merged. Recall that~\hindI{(\scites)} is polynomial-time solvable in
  this case~\cite{KA13}. We also achieve a (fixed-parameter) tractability result 
  for~\hindkM{}$(\scites)$ parameterized by the number~$k$ of article merges.
\begin{theorem}{}\label{lem:mergestrac}
If the compatibility graph~$G$ is a clique, then \hindkM{}$(\scites)$ is solvable in $O(9^kk \cdot (k+n + m))$~time, where $k$~is the number of allowed article merges.
\end{theorem}
\begin{proof}
  Assume that $(D, G, W, h, k)$ is a yes-instance and let~$\merge$ be
  a partition of~$W$ with \hind{} at least $h$ and at most $k$
  merges. Let~$M:=\{v\in W\mid v\in P,P\in\merge,|P|\geq 2\}$ be the
  set of articles that have been merged with other articles, and
  let~$W':=\{v\in W \mid \degi_D(v)\leq h\}$ be the set of articles with at
  most~$h$ citations. Let~$v_1, v_2, \ldots$ be the articles of~$W'$
  ordered by non-increasing citation counts. We claim that we may
  assume that~$M = \{v_1, \ldots, v_{|M|}\}$. Otherwise, we are in one
  of the following cases:
  \begin{enumerate}[{Case} 1.]
  \item There is an article~$v\in M$ with more than $h$~citations. That is, $v\in P\in\merge$ and $|P|\geq 2$. In this case, we may simply
    split~$P$ into~$P \setminus \{v\}$ and~$\{v\}$ without dropping
    the \hind{} of $\merge$ below $h$. 

  \item There is an article~$v_i\in M$ with $i>|M|$.  That is, $v_i
    \in P \in \merge$ with $|P|\geq 2$.  Then, since the compatibility graph is a clique, we may replace~$v_i$
    in~$P$ with an arbitrary article $v_j\notin M$ and $j \leq |M|$
    (which clearly exists) without decreasing the \hind{} of~$\merge$.
  \end{enumerate}
  Since at most $k$~article merges are allowed, we have $|M| \leq 2k$.
  Hence, if there is a solution, then there is also one where all
  merged articles are within~$\{v_1, \ldots, v_{2k}\}$.  Thus, we can
  remove all edges from the compatibility graph~$G$ that are incident
  with articles of at least~$h$ citations and discard all
  articles~$v_j$ with $j > 2k$.  In this way, we obtain an instance
  with a compatibility graph that contains at most~$2k$ vertices.  We
  now obtain the claimed fixed-parameter tractability result by
  adapting the dynamic programming algorithm behind
  \cref{lem:small-comp-fpt}.

  Since the only nontrivial connected component~\(C\) of the compatibility graph after the above preprocessing is a clique, we apply the algorithm only to~\(C\).  Thus, the auxiliary table~$Q$,
  used to store whether merging a set~$V'$ of articles creates an
  article with at least~$h$ citations, may ignore the compatibility
  graph. More formally, for all~\(V'\subseteq V(C)\), we let
  \[ 
  Q[V']:=
  \begin{cases}
    1&\text{if $\mu(V')\geq h$,}\\
    0&\text{otherwise.}
  \end{cases}
  \]
  In order to ensure that we make at most \(k\)~merges, we need an
  additional index in the main table~$T$. More precisely, for a set~\(V'\subseteq V(C)\) of vertices,
  let~$T[V',k]$ be the maximum number of parts~$P$ with $\mu(P)\geq h$
  in any partition of~$V'$ that can be obtained from the singleton
  partition by performing at most~$k$ merges.  Then,
  \[
  T[V',k]=
  \begin{cases}
    \enspace 0&\text{if $V'=\emptyset$},\\
    \enspace 0&\text{if $k \leq 0$, and}\\
\enspace\smashoperator{\max\limits_{V''\subsetneq V'}}(T[V'\setminus V'',k-(|V''|-1)]+Q[V''])&\text{otherwise.}
  \end{cases}
  \]
  Since the ground set~\(V(C)\) of articles considered in the dynamic programming table has size at most~$2k$, this algorithm has a running time of~$O(9^kk \cdot (k+n + m))$.
\end{proof}

\noindent For the~\ucites{} and~\mcites{} measure, we obtain hardness results for
both \hindkM{} and \hindI{}; the (parameterized) reductions are from the
\pIS{} problem.%
\decprob{\pIS}
{An undirected graph $H$ and a non-negative integer~$\ell$.}
{Is there an \emph{independent set} of size at least $\ell$ in~$H$, that is, a set of~$\ell$ pairwise nonadjacent vertices?}
\noindent\pIS{} is NP-hard and W[1]-hard with respect to~$\ell$ \cite{DF13}. %
\begin{theorem}\label{lem:mergeshardb}
  For  $\mu \in \{\ucites, \mcites\}$, \hindkM{}$(\mu)$ is W[1]-hard
  parameterized by $k$ even if the
  compatibility graph is a clique.
\end{theorem}

  \begin{proof}
    Let~$(H, \ell)$ be an instance of \pIS{}. We construct an
    instance~$(D, G, W, h, k := \ell-1)$ of \hindkM{} that is a
    yes-instance if and only if $(H,\ell)$~is a yes-instance for
    \pIS{}. Clearly, this is a parameterized reduction with respect to
    $\ell$ and $k$. 

    Let $n := |V(H)|$ and $h := \ell n$. Without loss of generality, we assume that $n > \ell > 1$. Our \hindkM{} instance has an article set~$W=W_\geq\uplus W_<$, where $W_<:=V(H)$ and $W_\geq$~consists of $h - 1$~new articles. Next, for each article~$w \in W_\geq$, we introduce $h$~new articles not in $W$ that cite~$w$ and no other article. The citations of the articles in~$W_<$ are defined as follows. For each pair of adjacent vertices~$u, v \in V(H)$, we introduce a new article~$e_{\{u, v\}}$ not contained in~$W$ that cites the articles~$u, v \in W_<$ and no other articles. Furthermore, we increase the citation counts of each article in~$W_<$ to exactly~$n$. That is, for each article~$w \in W_<$ we introduce new articles not contained in~$W$ that cite only~$w$ until~$w$ has $n$~citations. The compatibility graph~$G$ is a clique. This concludes the construction. 

    Clearly, the construction can be carried out in polynomial time.  Moreover, the reduction is a parameterized reduction from \pIS{} parameterized by~$\ell$ to \hindkM{} parameterized by~$k$ since~$k=\ell-1$. Note that no article in~$W$ cites another article in~$W$ and, hence, for any part~$P$ in a partition of~$W$, we have $\ucites(P) = \mcites(P)$. 

    Let us prove the correctness of the reduction. Assume first that~$(H, \ell)$ is a yes-instance and let~$S$ be an independent set of size~$\ell$ in~$H$. Then, merging all articles of~$S$ into one article in the \hindkM{} instance is valid since the compatibility graph~$G$ is a clique. Furthermore, it yields a merged article~$S$ with $\ucites(S) \geq h$~citations: Since the vertices in~$S$ are independent in~$G$, there is no article $e_{\{u, v\}}$ citing both~$u, v \in S$ in the \hindkM{} instance. Thus, the citations of the articles in~$S$ are pairwise disjoint. Together with the $h-1$ atomic articles in $W_\geq$ we have \hind{}~$h$.

\looseness=-1
    Conversely, assume that $(D, G, W, h, \ell - 1)$ is a yes-instance.  Since we are allowed to merge at most~$\ell - 1$ times in order to %
    achieve an \hind{} of~$h = \ell n$ and since each article in~$W_<$ has exactly $n$~citations, we need to merge $\ell$~articles of~$W_<$ into one article. That is, there is a part~$S \subseteq W_<$ in any solution for \hindkM{} with~$\ucites(S)\geq h$ citations. This means that the articles contained in~\(S\) have pairwise disjoint sets of citations because each of them has only $n = h/\ell$~citations. Thus, $S$~is an independent set in~$H$.
  \end{proof}

\noindent The reduction for
\cref{lem:mergeshardb} %
exploits the fact that at most $k$~merges are
allowed. Hence, to show NP-hardness for \hindI{}, we need a different
reduction.  Note that the NP-hardness for the \(\mcites\)~measure and general compatibility graphs already follows from \cref{qhard}\eqref{qhard3}.  We complement this result by the following theorem.

\begin{theorem}[]
  \hindI{}$(\mu)$ is NP-hard for $\mu \in \{\ucites, \mcites\}$ even
  if %
the compatibility graph is a clique.\label{thm:improvementhard}
\end{theorem}
\begin{proof}[Proof] %
\looseness=-1  We give a polynomial-time reduction from \pIS{}. %
  Let $(H, \ell)$ be an instance of \pIS{} and let $q :=
  |E(H)|$. Without loss of generality, we assume that $q \geq \ell >
  2$. We now construct an
  instance of \hindI{} with citation graph $D$, a set~$V$ of articles,
  and a subset~$W \subseteq V$ of own articles. The compatibility graph $G$ will be a
  clique on all articles. We introduce citations so that the \hind{}
  of the singleton partition of~$W$ will be~$q - 1$, hence the goal in
  the constructed instance will be to achieve \hind{} at least~$q$.

  The article set~$W$ is partitioned into three parts~$W = W_\geq \uplus\allowbreak W_{-1} \uplus W_<$. The first part, $W_\geq$, consists of $q - \ell - 1$ articles, and for each article~$w \in W_\geq$ we introduce $q$ articles not in~$W$ that cite~$w$ and no other article. The second part, $W_{-1}$, consists of $\ell$ articles, and for each article $w \in W_{-1}$ we introduce $q - 1$~articles not in~$W$ that cite~$w$ and no other article. The last part, $W_<$, contains the vertices of the \pIS{} instance, that is,~$W_<:=V(H)$. Finally, for each edge~$\{u, v\} \in E(H)$ we introduce one article $e_{\{u, v\}}$ not in $W$ that cites both $u$ and $v$. This concludes the construction of the citation graph~$D$. Note that the singleton partition of~$W$ has \hind{} $q - 1$. Hence, we have created an instance~$(D, G, W)$ of \hindI{} where we are looking to increase the \hind{} to at least~$q$. Clearly, we can carry out this construction in polynomial time. Furthermore, since there are no self-citations, that is, no articles in $W$ cite each other, for any subset~$P$ of~$W$ we have $\ucites(P) = \mcites(P)$.  Let us now prove the equivalence of the two instances.
    
    Assume that $(H, \ell)$ is a yes-instance. We claim that then also
    the \hindI{} instance is a yes-instance. Choose an independent
    set~$S$ of size~$\ell$ in~$H$. Take each of the corresponding
    articles in~$S$ and merge them with the articles in~$W_{-1}$,
    pairing them one by one. This creates $\ell$~merged articles with
    $q$~citations each. Together with the articles in~$W_\geq$, we now
    have~$q - 1$ articles with $q$~citations, some of them merged. To
    create another article with $q$~citations, simply merge all
    articles in~$W_< \setminus S$ into one article: Since~$S$ is an
    independent set, for each article~$e_{\{u, v\}}$ citing~$W_<$,
    either~$u$ or~$v$ is not in~$S$. Hence, the merged article~$W_<
    \setminus S$ has $q$~citations. Thus, $(D,G,W)$ is a yes-instance.
    
    Now assume that~$(D, G, W)$ is a yes-instance and let us show that
    also $(H, \ell)$ is. Take a partition $\merge$ of $W$ with \hind{}
    at least $q$. Note that any subset~$R \subseteq W_<$ has~$\mu(R)
    \geq q$ only if $R$ is a vertex cover of~$H$ (a vertex cover of a
    graph is a subset~$X$ of the vertices such that each edge is
    incident with some vertex in~$X$). Hence, as there are at most~$q
    - 1$ parts~$P \in \merge$ with~$P \not\subseteq W_<$, there is at
    least one part~$P \in \merge$ such that~$P \cap W_<$ is a
    vertex cover of~$H$. For the sake of contradiction, assume that there are two parts~$P_1,
    P_2$ such that $P_1 \cap W_<$ and $P_1 \cap W_<$ are vertex covers for~$H$. Then~$P_1 \cup P_2 \supseteq V(H) = W_<$. Furthermore, each remaining part
    of~$\merge$ contains only articles in~$W_\geq \cup W_{-1}$, that
    is, out of these parts, at most~$q - \ell - 1 + \lfloor\ell/2\rfloor$
    can have at least~$q$ citations. However, as~$\ell > 2$, there are
    at most~$q - \lceil\ell/2\rceil - 1 + 2 \leq q - 1$ parts with at
    least~$q$ citations in~$\merge$, a contradiction. Thus, there is exactly one
    part~$P \in \merge$ such that~$R \coloneqq P \cap W_<$ is a vertex cover of~$H$.
    
    Take $S := V(H) \setminus R$. Note that, since~$R$ is a vertex
    cover of~$H$, $S$~is an independent set in~$H$; we claim that~$S$
    has size at least~$\ell$. Since there is exactly one part
    in~$\merge$ that contains a vertex cover of~$H$, each remaining part has at least $q$~citations and there are at
    least~$q - 1$ of them. This means that no two articles in~$W_\geq
    \cup W_{-1}$ are merged. Hence, each article in~$W_{-1}$ is merged
    into an article in~$S$, that is, $S$~contains at least~$\ell$
    articles.
\end{proof}

\section{Achieving a moderately large \hind{}}\label{sec:parh}
\noindent
We now consider the \hind{} that we want to achieve as a parameter. This
parameter is often not very large as researchers in the early stage
of their career have an \hind{} below 20. Even for more experienced
researchers the \hind{} seldom exceeds 50. Hence, in many cases, the
value of a desired \hind{} is sufficiently low to serve as useful
parameter in terms of gaining efficient \fp{}
algorithms.  However,
note that \cref{corlem} immediately implies that \hindM{}$(\mu)$ is W[1]-hard with respect to the target \hind~$h$. This hardness also transfers to the $\ucites$ and $\mcites$ measures:  
\begin{corollary}\label{thm:hthm}
  \hindM{}$(\mu)$ is W[1]-hard with respect to the target \hind{}~\(h\) for each $\mu \in \{\scites, \ucites,\mcites\}$.
\end{corollary}
In contrast, we now show that \hindM{}($\mu$) is \fp{} tractable for any citation measure~\(\mu\in\{\scites,\ucites,\mcites\}\) if the compatibility graph is a clique. %
To this end, we describe a kernelization algorithm, that is, a polynomial-time data reduction algorithm that produces an equivalent instance whose size is upper-bounded by some function of the parameter~$h$. %

\looseness=-1
One particular difficulty in designing data reduction rules for the \(\mcites\) measure is that citations in~$D[W]$ are somewhat fragile as they may be ``destroyed'', for example, if two adjacent vertices in~$W$ are merged.
Thus, we take the following route to obtain the problem kernel. First, in \(O(n+m)\)~time, we use a greedy strategy to compute a maximal matching in the undirected graph underlying the citation graph~\(D\).  If this matching has size at least \(h^2\), then we show that there is a partition achieving \hind{}~$h$.  Otherwise, we use the fact that the articles that do not participate in the matching do not cite each other to design further data reduction rules.

\begin{rrule}\label{rule:matching}
  Let~$(D,G,W,h)$ be an instance of \hindM{}$(\mu)$ for \(\mu\in\{\scites,\ucites,\allowbreak\mcites\}\) such that~$G$
  is a clique. Compute a maximal matching in~$D$ by iteratively
  putting an arc into the matching as long as possible. If the
  resulting matching has size at least~$h^2$, then accept.
\end{rrule}

\begin{lemma}
  \cref{rule:matching} is correct and an application can be performed in $O(n + m)$~time.
\end{lemma}
\begin{proof}
  Let $M$ be a matching of size~$h^2$
  in~$D$. Let~$W'$ denote the set of~$h^2$
  vertices that are the heads of the arcs in~$M$. The articles in~$W'$ are
  cited by the tails of the respective arcs
  in~$M$. Thus, we may assume $W'\subseteq W$
  since only citations to vertices in~\(W\) are counted.  Consider a
  partition~$\merge$ of~$W$ that is obtained by partitioning~$W'$
  into exactly~$h$ sets, each of size~$h$, and not merging any other
  articles in~$W$. Since~$M$ is a matching, there are, for each merged
  article~$P\in \merge$, at least $h$~independent arcs from an article
  in~$V\setminus W'$ to an article in~$P$.
  Thus, \(\merge\)~has \hind{}~\(h\) with respect to \(\scites\) and \(\ucites\).
  Moreover, since the articles
  in~$V\setminus W'$ are not merged, there are thus $h$~distinct
  unmerged articles that cite an article of~$P$. Hence, $\merge$~has
  \hind{}~$h$ with respect to~\(\mcites\).

  To see the claim about the running time, observe that it suffices to
  iterate once over all edges, maintaining a label for each vertex
  that indicates whether it has an incident edge in the matching.
\end{proof}
\noindent Now assume that we have computed a maximal matching of size at most~\(h^2\) in~$D$ in \(O(n+m)\)~time. Then, the vertices incident to the matching arcs form a vertex cover~\(C\) of size at most \(2h^2\) for~$D$. It remains to upper bound the number of articles in the independent set~$V \setminus C$.  To this end, we first give a data reduction rule that ensures that each article in~$C$ cites only few articles in~$W \setminus C$. To do this, we need the following lemma, which enables us to assume that a solution merges only few articles in the independent set.

\begin{lemma}\label{lem:few-merges-mcites}
  Let $(D,G,W,h)$ be a yes-instance of \hindM{}$(\mu)$,
  where \(\mu\in\{\scites,\ucites,\allowbreak\mcites\}\),
  such that $G$~is a clique and let $X \subseteq W$ such that
  no article in~$X$ cites any other article in~$X$.
  Then, there is a partition~$\merge$ of~$W$ that has \hind{} at least~$h$
  with respect to~\(\mu\)
  and such that at most~$h^2$ atomic articles from~$X$
  are not singletons in~$\merge$.
\end{lemma}
\begin{proof}
  \looseness=-1
  Consider a partition~$\merge^*$ that has \hind{} at least~$h$ 
  with respect to~\(\mu\) and such that the number of atomic
  articles from~$X$ that are not singletons in~$\merge^*$ is minimum.  
  If $\merge^*$~has more than~$h$ merged articles
  or a merged article~\(P\) with \(\mu(P)<h\), then 
  one of the merged articles can be split into its atomic articles
  without decreasing the \hind{} below~$h$.
  Thus, assume that $\merge^*$~contains at most~$h$ merged articles~\(P\),
  each with~\(\mu(P)\geq h\).
  To prove the lemma, it is enough to show that there is no
  merged article~$P^* \in \merge^*$
  containing more than~$h$ atomic articles~$a_1,\ldots, a_{\ell}\in X$,
  where~$\ell > h$.
  We will lead the existence of such an article~\(P^*\) to a contradiction.

  If \(\mu=\scites\), then it is clear that \(P^*\)~can be replaced
  by two articles~\(P^*\setminus\{a_i\},\{a_i\}\) in~\(\merge^*\)
  for some~\(i\in\{1,\dots,\ell\}\) 
  so that still \(\scites(P^*\setminus\{a_i\})\geq h\),
  contradicting our choice of~\(\merge^*\).

  For \(\mu=\ucites\), consider the
  articles~$P_j \coloneqq (P^* \setminus X) \cup \{a_1, \ldots, a_j\}$
  for \(j\in\{1,\dots,\ell\}\).
  If \(\ucites(P_j)<\ucites(P_{j+1})\) for each~\(j\in\{1,\dots,\ell-1\}\),
  then \(\ucites(P_h)\geq h\) and replacing \(P^*\) by \(P_h\) in~\(\merge^*\)
  and adding the articles of \(P^*\setminus P_h\) as singletons
  yields a partition with more singletons from~\(X\) and \hind{} at least~\(h\),
  thus contradicting our choice of~\(\merge^*\).  
  Otherwise, there is a $j \in \{1, \ldots, \ell - 1\}$
  such that $\ucites(P_j) \geq \ucites(P_{j+1})$, meaning that
  \(a_{j+1}\) is cited by a subset of the atomic articles that cite~\(P_j\).
  Then, replacing \(P^*\)~by the two articles~\(P^*\setminus\{a_{j+1}\}\)
  and \(\{a_{j+1}\}\) in~\(\merge^*\) yields a partition with \hind{} at
  least~\(h\) and more singletons from~\(X\),
  contradicting our choice of~\(\merge^*\).

  Similarly, for \(\mu=\mcites\), consider the series of partitions~$\merge_j$
  for $j \in \{1, \ldots, \ell - 1\}$ arising from
  replacing~$P^*$ by the $\ell - j + 1$~(merged)
  articles~$P_j$ and $\{a_{j + 1}\}, \ldots, \{a_{\ell}\}$ in~\(\merge^*\).
  If $\mcites_{\merge_j}(P_j) < \mcites_{\merge_{j+ 1}}(P_{j+1})$
  for each $j \in \{1, \ldots, \ell - 1\}$,
  then $\mcites_{\merge_h}(P_h)\geq h$
  and each remaining merged article~\(P'\in\merge_h\cap\merge^*\)
  has \(\mcites_{\merge_{h}}(P')\geq\mcites_{\merge}(P')\).
  Yet $\merge_h$~has less atomic articles from~$X$ that are not singletons
  than~\(\merge^*\), which contradicts our choice of $\merge^*$.
  If, otherwise, there is a $j \in \{1, \ldots, \ell - 1\}$
  such that $\mcites_{\merge_j}(P_j) \geq \mcites_{\merge_{j+1}}(P_{j+1})$, 
  then we consider two cases.

  \looseness=-1
  In the first case,
  $a_{j + 1}$~does not cite any article in $P^* \setminus X$.
  Since articles in~\(X\) do not cite each other and \(a_{j+1}\in X\),
  it follows that \(a_{j+1}\)~does not cite any article in~\(P^*\).
  Therefore, $\mcites_{\merge_j}(P_j) \geq \mcites_{\merge_{j+1}}(P_{j+1})$
  implies that
  the set of (merged) articles that cite $a_{j + 1}$ is
  a subset of the (merged) articles that cite~$P_{j}$.
  This implies that the set of (merged) articles that cite $a_{j + 1}$ is
  a subset of the (merged) articles that cite~$P_{\ell}$.
  Thus, replacing the merged article~$P^* = P_{\ell}$
  by the two articles~$P^*\setminus \{a_j\}$ and $\{a_j\}$ in~$\merge$,
  we obtain a partition with \hind\ at least~$h$ that
  has one less atomic article from~$X$ that is not a singleton.
  This contradicts our choice of~$\merge^*$.

  In the second case,
  $a_{j + 1}$ cites at least one article in $P^* \setminus X$.
  Let $c$~be the number of (merged) articles outside of~$P_j$ that
  cite~$a_{j + 1}$ and none of the articles in~$P_j$ and
  let $d$~be the number of (merged) articles outside of~$P_j$ that
  cite both~$a_{j+1}$ and an article in~$P_j$.
  We have $c \leq d + 1$ since, otherwise,
  $\mcites_{\merge_j}(P_j) < \mcites_{\merge_{j+1}}(P_{j+1})$.
  Since no article in~\(X\) cites~\(a_{j+1}\) and
  \(P_\ell\setminus P_j\subseteq X\), we also have $c' \leq d' + 1$,
  where $c'$ is the number of (merged) articles outside of~$P_\ell$ that
  cite~$a_{j + 1}$ and none of the articles in~$P_\ell \setminus \{a_{j + 1}\}$ 
  and where $d'$ is the number of (merged) articles outside of $P_\ell$ that 
  cite both~$a_{j + 1}$ and an article in~$P_{\ell} \setminus \{a_{j + 1}\}$.
  Replacing the merged article~$P_{\ell}$
  by the two articles $P_{\ell} \setminus \{a_j\}$ and $\{a_j\}$ in~$\merge^*$,
  we obtain a partition with \hind\ at least~$h$ that has one less
  atomic article from~$X$ that is not a singleton.
  This is a contradiction to the choice of~$\merge^*$.
\end{proof}

\noindent The idea for the following data reduction rule is that, if an article in a vertex cover~$C$ cites many articles outside of~$C$, then only few of these are in merged articles and only few of the remaining articles are needed to maintain the citations of the merged articles. Hence, superfluous citations can be removed.

\begin{rrule}\label{rule:many-cites-mcites}
  Let $(D, G, W, h)$ be an instance of \hindM{}$(\mu)$,
  where \(\mu\in\{\scites,\allowbreak\ucites,\allowbreak\mcites\}\), such that $G$~is a clique and let $C$~be a vertex cover in~$D$. If there is an article~$v \in C$ that cites more than $2h^2 + 2h$~articles in~$W \setminus C$, then remove an arbitrary citation~$(v, w)$ for some $w \in W \setminus C$.
\end{rrule}
\begin{lemma}
  \cref{rule:many-cites-mcites} is correct and can be exhaustively applied in $O(n + m)$ time.
\end{lemma}
\begin{proof}
  We first prove the correctness.
  Clearly, if the instance
  resulting from an application of \cref{rule:many-cites-mcites}
  is a yes-instance,
  then also the original instance is a yes-instance.
  For the converse, consider a partition~$\merge$
  with \hind{}~$h$ for $(D, G, W, h)$
  that does not have \hind~$h$ after removing~$(v, w)$ from~$D$.
  Without loss of generality, we can assume that
  there are no merged articles~$P \in \merge$ with $\mu(P) < h$,
  because unmerging such articles can only
  increase the number of citations of other articles (in the case of \mcites{})
  and, hence, cannot decrease the \hind{} of~$\merge$.

  Observe that, after the deletion of citation~\((v,w)\) from~\(D\),
  there is at most one merged article~$P \in \merge$ such that~$\mu(P) < h$.
  We claim that, among the articles cited by~$v$,
  there is an atomic article with less than \(h\)~citations
  that we can add to~$P$
  so that \(P\)~has $h$~citations again,
  thus yielding a partition with \hind{} at least~\(h\).
  Let~$U$ denote
  the set of more than \(2h^2+2h\)~articles in~$W \setminus C$
  that are cited by~$v$.
  By \cref{lem:few-merges-mcites}, we may assume that
  at most~$h^2$ of the articles in~$U$ are in merged articles.
  Thus, since \(\merge\)~does not have \hind~\(h\),
  there is an article~\(u\in U\) that is a singleton in~\(\merge\)
  and satisfies \(\mu(\{u\})<h\).
  If \(\mu=\scites\), then adding~\(u\) to~\(P\)
  yields a partition with \hind{}~\(h\)
  because \(u\)~has at least one citation (from~\(v\)).
  If \(\mu=\ucites\), then observe that $v$ does not cite~$P$. Hence, adding~\(u\) to~\(P\)
  yields a partition with \hind{}~\(h\)
  since \(v\) cites~\(P\cup\{u\}\).
  It remains to handle the case \(\mu=\mcites\).
  
  Note that having \(\mcites_{\merge}(P)<h\)
  after deleting~\((v,w)\) from~\(D\) means \(v\notin P\).
  Recall that, by \cref{lem:few-merges-mcites}, we may assume that
  at most~$h^2$ of the articles in~$U$ are in merged articles.
  Furthermore, there are at most~$h - 1$ articles in~\(U\) that cite~$P$ and at most $h - 1$ articles~$u \in U$ with~$\mcites_\merge(u) \geq h$.
  Denote the remaining articles of~$U$ by~$u_1, \ldots, u_\ell$. That is, each $u_i$ is cited by~$v$, is a singleton in $\merge$, does not cite~$P$, and has $\mcites_\merge(u_i) < h$.
  Observe that,
  if one of these articles, say~$u_i$, does not cite
  any merged article in~$\merge$,
  then adding~$u_i$ to~$P$ yields a partition with \hind{}~$h$.
  Hence, assume that each article~$u_i$ for $i \in \{1, \ldots, \ell\}$
  cites at least one merged article.
  Observe furthermore that,
  if there is some~$u_i$ such that each merged article~$P' \neq P$
  that is cited by~$u_i$
  receives~$h$ citations from~$u_1, \ldots, u_{i - 1}$,
  then adding~$u_i$ to~$P$ yields a partition with \hind{}~$h$
  (recall that $u_i$ does not cite~$P$).
  Call such an article~$u_i$ \emph{good}. We claim that there is at least one good article. Assign to each $u_i$ the integer~$c_i \coloneqq \sum_{P' \in \merge \setminus \{P\}} \min\{h, \operatorname{cites}(i, P')\}$, where $\operatorname{cites}(i, P')$ is the number of citations of~$P'$ from~$u_1, \ldots, u_{i - 1}$. Observe that each~$u_i$ either cites at least one merged article~$P' \neq P$ that receives less than~$h + 1$ citations from~$u_1, \ldots, u_{i - 1}$ or it is good. Hence, either $c_i > c_{i - 1}$ or $u_i$ is good. Furthermore, if $c_i \geq (h - 1)h$, then $u_i$ is good. Thus, if $\ell > (h - 1)h$, then there is a good article. Because $\ell \geq |U| - h^2 - (h - 1) - (h - 1)$ and $|U| \geq 2h^2 + 2h$, there is a good article indeed.

  Regarding the running time, for each article~\(v\in V\), it can be checked in \(O(\dego_D(v))\)~time whether the rule is applicable. In the same time, all citations exceeding the number \(2h^2 + 2h\)~can be deleted.  Since application of the reduction rule to one article cannot make it applicable to other articles, it follows that it can exhaustively be applied in \(O(\sum_{v\in V}\dego_D(v))=O(n+m)\)~time.
\end{proof}
Finally, we need the following cleanup rule.
\begin{rrule}\label{rule:cleanup-mcites}
  Let $(D, G, W, h)$ be an instance of \hindM{}$(\mu)$,
  where \(\mu\in\{\scites,\allowbreak\ucites,\allowbreak\mcites\}\),
  such that $G$~is a clique.
  \begin{itemize}
  \item If there are $h$ articles in~$W$ with $h$ citations each, then accept.
  \item If there is an article in~$W$ that is not cited, then remove this article.
  \item If there is an article in~$V \setminus W$ that cites no other article, then remove this article.
  \item If there is an article in $W$ that is cited by more than~$h$
    articles in $V \setminus W$, then remove an arbitrary one of
    these citations.
  \end{itemize}
\end{rrule}
\begin{lemma}
  \cref{rule:cleanup-mcites} is correct and can be exhaustively applied in $O(n + m)$~time.
\end{lemma}
\begin{proof}
  It is clear that \cref{rule:cleanup-mcites} can be exhaustively applied in $O(n + m)$~time. For the correctness, the only non-obvious part is the last one. To see that it is correct, let $v$ be an article to which it has been applied, and observe that every merged article that~$v$ can be contained in has at least~$h$ citations before applying the rule, as well as after applying the rule.  
\end{proof}

\noindent Combining all data reduction rules above, we can give the promised polynomial-size problem kernel.

\begin{theorem}\label{thm:kern-fusion}
  If the compatibility graph is a clique, then
  a $(4h^4 + 6h^3 + 5h^2)$-article problem kernel
  for \hindM$(\mu)$
  is computable in \(O(n+m)\)~time,
  where \(\mu\in\{\scites,\allowbreak\ucites,\allowbreak\mcites\}\).
\end{theorem}
\begin{proof}
  To compute the problem kernel, apply exhaustively
  \cref{rule:matching,rule:many-cites-mcites,rule:cleanup-mcites}. By
  the corresponding lemmas, the resulting instance is a yes-instance if and only
  if the input instance is a yes-instance, and the rules can be carried out in $O(n + m)$~time.

  To see the upper bound on the size, let~$C$ be the vertex cover of~$D$ computed from the matching of \cref{rule:matching}. Note that $|C| \leq 2h^2$. We upper bound the size of~$W$ and, due to reducedness with respect to \cref{rule:cleanup-mcites}, it then suffices to upper bound the number of articles that cite or are cited by articles in~$C$. We divide the articles into four groups:
  \begin{itemize}
  \item The set~$W_{\geq} \subseteq W$ of articles with at least~$h$
    citations from articles in~$V \setminus W$, 
  \item the set~$W_{<} \subseteq W$ of articles with less than~$h$ citations from articles in~$V \setminus W$,
  \item the set~$V_{<} \subseteq V \setminus (W \cup C)$ of articles that cite articles in~$W_{<}$, and
  \item the set~$V_{\geq} \subseteq V \setminus (W \cup C \cup V_{<})$ of articles that cite articles in~$W_{\geq}$ but no article in~$W_{<}$.
  \end{itemize}
  Clearly, $V = C \cup W_{\geq} \cup W_{<} \cup V_{<} \cup V_{\geq}$. To upper bound the size of~$V$, first note that $|W_{\geq}| \leq h - 1$ by \cref{rule:cleanup-mcites}. For $W_{<}$, note that each of these articles is either contained in~$C$ or cited by at least one article in~$C$. By reducedness with respect to \cref{rule:many-cites-mcites}, there are hence at most $2h^2 + (2h^2 + 2h)|C| \leq 4h^4 + 4h^3 + 2h^2$~articles in~$W_{<}$. Since each article in $V_{<}$ cites at least one article in~$W_{<} \cap C$, and since these articles receive at most~$h - 1$ such citations each, there are at most $2h^3$ articles in~$V_{<}$. Finally, each article in $V_{\geq}$ cites at least one article in~$W_{\geq}$, and these articles receive at most~$h$ citations from articles in~$V \setminus W$ each by \cref{rule:cleanup-mcites}. Thus, there are at most~$h^2$ articles in~$V_{\geq}$ and, overall, there are at most $$|C| + |W_{\geq}| + |W_{<}| + |V_{<}| + |V_{\geq}| \leq 2h^2 + (4h^4 + 4h^3 + 2h^2) + 2h^3 + h^2 = 4h^4 + 6h^3 + 5h^2$$ articles in a reduced instance.
\end{proof}

\noindent Applying an arbitrary algorithm that decides \hindM{} to the instances resulting from the problem kernel in \cref{thm:kern-fusion} now yields the following classification result.
\begin{corollary}\label{lem:kernel-h}
  If the compatibility graph is a clique, then \hindM{}$(\mu)$ with~$\mu\in \{\scites,\allowbreak\ucites,\allowbreak\mcites\}$ is linear-time solvable for constant~\(h\).
\end{corollary}

\section{Experiments}
\label{sec:exp}

\noindent\looseness=-1
In this section, we examine by how much authors can increase their
\hinds{} when allowing only merges of articles with similar titles or
when fixing the allowed number of merges. To this end, we gathered
article and citation data of AI researchers, computed compatibility
graphs based on similarity of article titles, and implemented
heuristics and exact algorithms for maximizing the \hind{}. Herein, we
focused on the islands of tractability that we determined in our
theoretical analysis, that is, the cases of small number of merges and
small connected components in the compatibility graph for the measures
$\scites$ and $\ucites$. These cases are also practically relevant, as
$\ucites$ is the measure used by Google Scholar. The implemented
algorithms are mainly based on
\cref{lem:small-comp-fpt,lem:mergestrac}.

\paragraph{Data acquisition} We crawled Google Scholar data of
22~selected authors of IJCAI'13.  Our (biased) selection was based on capturing
authors in their early career, for whom \hind{} manipulation would seem
most attractive. Specifically, we selected authors who have a Google
Scholar profile, an \hind{} between 8 and~20, between 100 and 1000
citations, who are active between 5 and 10 years, and do not
have a professor position. 

In addition, we crawled Google Scholar data of \emph{AI's 10 to Watch}, a list of young accomplished researchers in AI that is compiled every two years by \textit{IEEE Intelligent Systems}. The dataset contains five profiles from the 2011 and eight profiles from the 2013 edition of the list~\cite{ai10tw11,ai10tw13}. Some profiles were omitted due to difficulties in the crawling process, for example, because of articles that could not be attributed unambiguously to the respective author due to non-unique author names. Compared to the IJCAI 2013 author set, AI's 10 To Watch 2011 contains researchers who are more experienced and AI's 10 To Watch 2013 falls in between these two data sets in this regard. \Cref{tab:data-sets} gives an overview of the properties of the data sets.
\begin{table}[t]\centering
  \caption{Properties of the three data sets. Here, $p$~is the number of profiles for each data set, $\overline{|W|}$~is the average number of atomic articles, $\overline{c}$~is the average number of citations, $\overline{h}$~is the average \hind{} in the data set, and \(h/a\)~is the average (unmanipulated) \Hind{} increase per year; the `$\mathrm{max}$' subscript denotes the maximum of these values.}
  \begin{tabular}[t]{@{}rrrrrrrrrr@{}}
\toprule
 & \multicolumn{1}{c}{$p$} & \multicolumn{1}{c}{$\overline{|W|}$}  & \multicolumn{1}{c}{$|W|_{\mathrm{max}}$} & \multicolumn{1}{c}{$\overline{c}$}  & \multicolumn{1}{c}{$c_{\mathrm{max}}$} & \multicolumn{1}{c}{$\overline{h}$}  & \multicolumn{1}{c}{$h_{\mathrm{max}}$} &\(h/a\) \\ \midrule

AI's 10 To Watch 2011 & 5 & 170.2 & 234 & 1614.2 & 3725 & 34.8 & 46 & 2.53\\
AI's 10 To Watch 2013 & 8 & 58.25 & 144 & 542.0 & 1646 & 14.0 & 26 & 2.77\\
IJCAI 2013 & 22 & 45.91 & 98 & 251.5 & 547 & 10.36 & 16 & 1.24 \\ \bottomrule
\end{tabular}
\label{tab:data-sets}
\end{table}

\looseness=-1 For each author, we computed upper and lower bounds for the \hind{} increase when allowing at most $k = 1,\dots,12$~merges of arbitrary articles and the maximum possible \hind{} increase when merging only articles whose titles have a similarity above a certain compatibility threshold $t = 0.1,\allowbreak 0.2,\dots,0.9$. The compatibility thresholding is described in more detail below.

\paragraph{Generating compatibility graphs} Compatibility graphs are
constructed using the following simplified `bag of words model': Compute
for each article~$u$ the set of words~$T(u)$ in its title. Draw an
edge between articles~$u$ and~$v$ if~$|T(u)\cap T(v)|\ge t\cdot
|T(u)\cup T(v)|$, where~$t\in [0,1]$ is the \emph{compatibility
  threshold}. For~$t=0$, the compatibility graph is a clique.  For~$t=1$,
only articles with the same words in the title are adjacent. Inspection showed that, for~$t\le 0.3$,
already very dissimilar articles are considered compatible.

\paragraph{Implemented algorithms}
We implemented our algorithms for the parameter ``maximum connected component size~$c$ of the compatibility graph'' (\cref{lem:small-comp-fpt}) and for the parameter~$k$ of allowed merges (\cref{lem:mergestrac}). We ran both algorithms using both the $\scites{}$ and $\ucites{}$ measures.  Note that, when applied with the \(\ucites{}\) measure, the algorithm for \cref{lem:mergestrac} does not necessarily compute the maximum possible \hind{} increase (cf.\ \cref{lem:mergeshardb}), but we note that it yields a lower bound. Moreover, running it with $\scites{}$ yields an upper bound for the maximum achievable \Hind{} with~$\ucites{}$ and thus, we obtain both a lower and upper bound on the achievable \hind{} with respect to~\(\ucites{}\) using \(k\)~merges.

The fixed-parameter algorithm
for parameter~$c$, the size of the connected components of the compatibility
graph, is not able to solve all instances. In particular, it fails
for~$t=0.2$, where it runs out of memory in most cases. We thus implemented
an alternative solution strategy that is based on the enumeration of
cliques in the compatibility graph, exploiting the fact that any
merged article is a clique in the compatibility graph~$G$. Thus, a
partition of the article set~$W$ that complies with~$G$ directly
corresponds to a set of vertex-disjoint cliques in~$G$.

Starting with~$h=1$, we do the following.
\begin{enumerate}
\item Enumerate all minimal sets~$P$ such that~$P$ is a clique
  in the compatibility graph and~$\mu(P)>h$. Each set~$P$ is a
  \emph{potential merged article} in a merged profile that
  achieves~\hind{}~$h$; clearly, we can restrict attention to minimal
  sets.
\item Find a maximum-cardinality set~$\merge'$ of potential merged
  articles such that~$P\cap P'=\emptyset$ for each pair~$P,P'\in
  \merge'$.
\item If~$|\merge'|>h$, then an~\hind{} of at least~$h$ can be
  achieved via merging. Continue with~$h\leftarrow h+1$. Otherwise, an~\hind{} of~$h$ cannot be achieved; return~$h-1$ as the maximum~\hind{} that can be achieved via merging.
\end{enumerate}
In the implementation of Step~1, we first enumerate all maximal
cliques of the compatibility graph and then check for each subset of
each maximal clique whether it is a minimal set such that~$\mu(P)>h$.
In Step~2, the size of~$\merge'$ is computed by constructing an
auxiliary graph whose vertices are the potential merged articles and
where edges are added between potential merged articles that have
nonempty intersection. In this graph, $\merge'$~is a
maximum-cardinality independent set. We compute~$\merge'$ by computing
a minimum-cardinality vertex cover via a simple fixed-parameter
algorithm for the parameter vertex cover size.

This algorithm has a higher worst-case running time than the
fixed-parameter algorithm for parameter~$c$: the overall number of
potential merged articles~$p$ may be exponential in~$c$ and we solve
\textsc{Independent Set} on a graph of order~$p$. Nevertheless, it
works for the three data sets as the number of potential merged
articles is much lower than in a worst-case instance.

Source code and data are freely available at \url{http://fpt.akt.tu-berlin.de/hindex}.

\begin{figure}[t]
  \centering
  \begin{subfigure}{8cm}
    \includegraphics{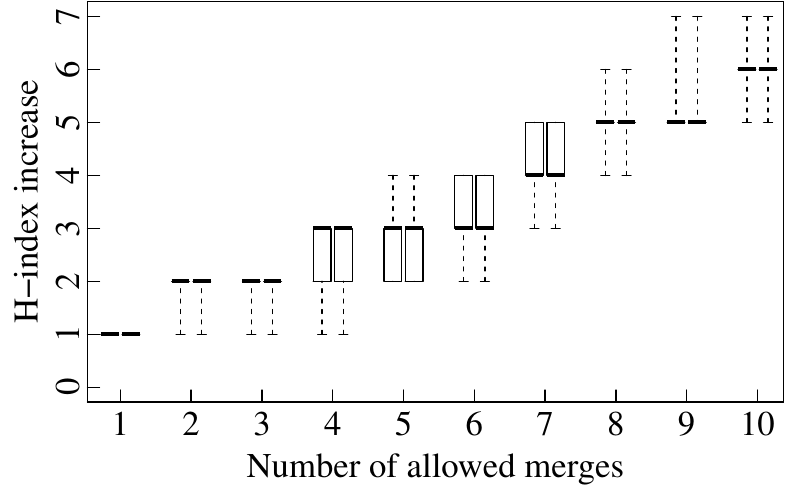}
    \caption{AI's 10 to watch 2011}
  \end{subfigure}
  \begin{subfigure}{8cm}
    \includegraphics{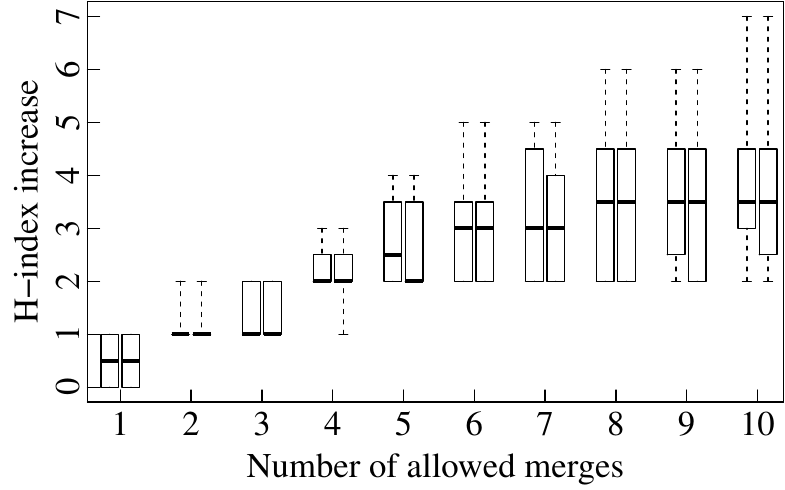}
    \caption{AI's 10 to watch 2013}
  \end{subfigure}
  \begin{subfigure}{8cm}
    \includegraphics{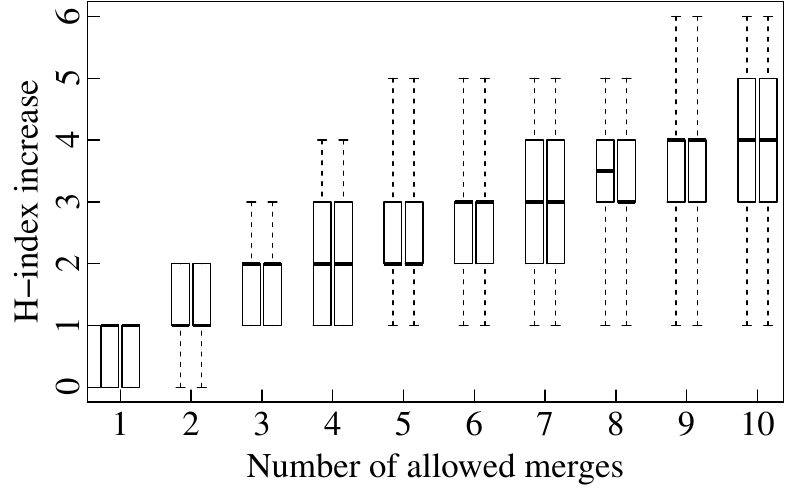}
    \caption{IJCAI'13 authors}
  \end{subfigure}
  \caption{Box plots for our three data sets of the achievable \hind{} increases when the number of merges is restricted but arbitrary pairs of articles can be merged (that is, the compatibility graph is a clique). For each number~$k$ of allowed merges, the left box shows the \hind{} increase for $\scites{}$, the right box shows lower bounds on the possible H-index increase for~$\ucites{}$. The lower edge of a box is the 25th percentile and the upper edge is the 75th percentile, a thick bar is the median. The whiskers above and below each box extend to the maximum and minimum observed values.}
  \label{fig:merges}
\end{figure}

\paragraph{Experimental results}  

We ran our algorithms under a time limit of one hour
on a 3.6\,GHz Intel Xeon E5-1620 processor
and a memory limit of 64\,GB.
Under these limits, the fixed-parameter algorithm for parameter~\(k\), the number of allowed merges,
failed to solve instances with~\(k\geq 11\).
Thus, \cref{fig:merges} shows results for~\(k\leq 10\) only.
The fixed-parameter algorithm for parameter~\(c\), the size of the connected components of the compatibility graph,
failed to solve instances with a compatibility threshold~$t\leq 0.2$.
Instances with $k\leq 10$ and~$t\geq 0.3$ were usually solved
within few seconds and using at most 100\,MB of memory.
The algorithm based on clique enumeration
solved each instance with~$t\geq 0.2$ in several minutes.
None of our algorithms was able to solve all instances with~$t=0.1$.
Thus, \cref{fig:comp} shows results for~\(t\geq 0.2\) only.

\cref{fig:merges} shows the \hind{} increase over all authors for each number~$k = 1,\dots, 10$ of allowed article merges when the compatibility graph is a clique.  Remarkably, three~merges are sufficient for all of our sample authors to increase their \hind{} by at least one.  Let us put this number into perspective: as shown in \cref{tab:data-sets}, we measured that, without manipulation, on average the \hind{} in each group of our sample authors grows between 1.24 and 2.77 per year (which is higher than the one-per-year increase observed by \citet{Hir05} in physics).  Thus, from \cref{fig:merges}, one can conclude that two merges could save about nine months of work for half of our AI's 10 To Watch 2011 group, about four months of work for half of our AI's 10 To Watch 2013 group, and 19 months of work for half of our IJCAI'13 group.

\begin{figure}[t]
  \centering
  \begin{subfigure}{8cm}
    \includegraphics{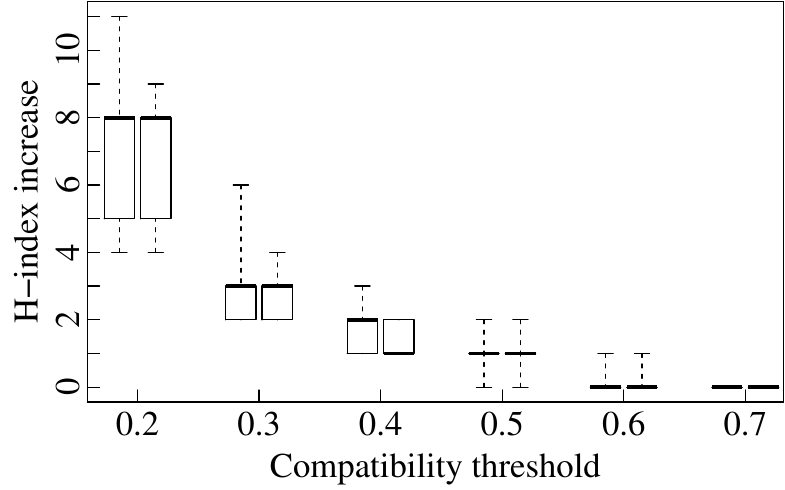}
    \caption{AI's 10 to watch 2011}
  \end{subfigure}
  \begin{subfigure}{8cm}
    \includegraphics{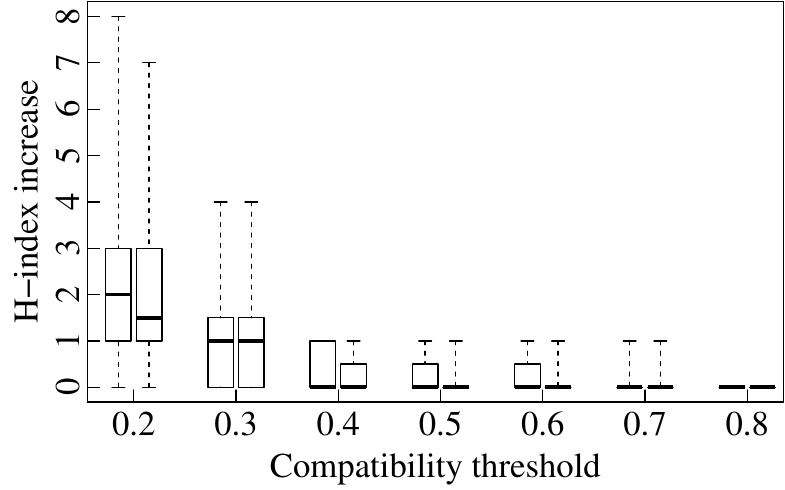}
    \caption{AI's 10 to watch 2013}
  \end{subfigure}
  \begin{subfigure}{8cm}
    \includegraphics{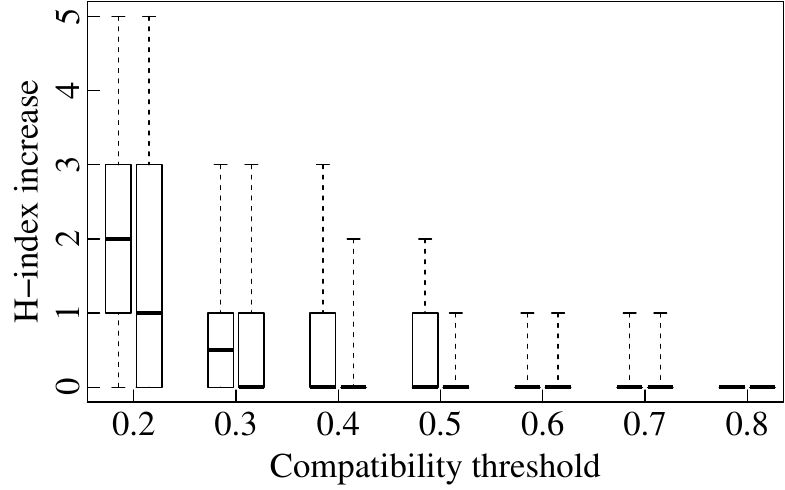}
    \caption{IJCAI'13 authors}
  \end{subfigure}
  \caption{Box plots for our three data sets of the achievable \hind{} increases when compatibility of articles is restricted but an arbitrary number of articles can be merged. For each compatibility threshold~$t$, the left box shows
    the \hind{} increase for $\scites{}$, the right box
    for~$\ucites{}$. The lower edge of a box is the 25th percentile
    and the upper edge is the 75th percentile, a thick bar is the
    median. The whiskers above and below each box extend to the
    maximum and minimum observed values.}
  \label{fig:comp}
\end{figure}

\cref{fig:comp} shows the \hind{} increase over all authors
for~$\ucites$ and each compatibility threshold~$t = 0.2,\allowbreak
0.3,\allowbreak\dots,0.9$.  Remarkably, when using a compatibility
threshold~$t\geq0.6$, $75\%$ of our sample authors cannot increase
their \hind{} by merging compatible articles.  We conclude that increasing the
\hind{} substantially by article merges should be easy
to  discover %
since it is necessary to merge articles with highly
dissimilar titles for such a manipulation.

\section{Outlook}\label{sec:conclusion}
\noindent
Clearly, it is interesting to consider merging articles in order to
increase other measures than the \hind{}, like the $g$-index
\cite{Egghe2006,Woe08b}, the $w$-index \cite{Woe08}, or the
$i10$-index of a certain author. The $i10$-index, the number of
articles with at least ten citations, is also currently used by Google
Scholar. Elkind and Pavlou~\cite{EP16} performed a study in this 
direction and, among other results, showed that the $g$-index and the 
$i10$-index seem somewhat easier to manipulate than the \hind{}. %
In addition, they also studied a scenario where the manipulator wants to take 
into account the impact of the manipulation actions on other researchers
(distinguishing between friends and competitors).

Moreover, merging articles in order to increase one index might
decrease other indices, like the overall number of citations. Hence,
it is also interesting to study the problem of increasing the \hind{} by merging
without decreasing the overall number of citations or the $i10$-index
below a predefined threshold. A systematic study of computing 
Pareto optimal solutions could also be interesting.

The computational problems related to optimal merging of articles in
the different measures are quite natural as evidenced for example by
their relation to~\textsc{Bin Covering} and \textsc{Machine
  Covering}. Thus, improvements over the presented algorithms
would be desirable as well as a study of further parameterizations in
a broad multivariate complexity analysis~\cite{FJR13,KN12,Nie10}.
 
Altogether, our experiments show that the merging option leaves some
room for manipulation but that \emph{substantial} manipulation
requires merging visibly unrelated articles. Hiring committees that
use the \hind{} in their evaluation thus should either examine the
article merges more closely or rely on databases that do not allow
article merges.

\section*{Acknowledgments}
\noindent We thank the anonymous referees for their helpful comments, in particular for pointing out \cref{qhard}\eqref{qhard3}.

René van Bevern was supported by the German Research Foundation (DFG), project DAPA (NI 369/12), at TU Berlin and by the Russian Foundation for Basic Research (RFBR), project~16-31-60007 mol\textunderscore{}a\textunderscore{}dk, at Novosibirsk State University.  Christian Komusiewicz was supported by the DFG project~MAGZ (KO~3669/4-1) and Manuel Sorge was supported by the DFG project DAPA (NI 369/12). Toby Walsh was supported by the Alexander von Humboldt Foundation, Bonn, Germany, while at TU~Berlin. The main work was done while Toby Walsh was affiliated with University of New South Wales and Data61, Sydney, Australia.

\bibliographystyle{abbrvnat}
\bibliography{hindex}
\end{document} 

